\documentclass[10pt,twocolumn,final]{IEEEtran}

\usepackage{amsmath,amsfonts,algorithm,algorithmic,subfig,url}
\usepackage[amsmath,thmmarks]{ntheorem}
\usepackage{graphicx}

\def\RN{ \mathbb{R} }                               % real numbers
\newcommand{\EE}{{\mathbb E}}
\newcommand{\PP}{{\mathbb P}}
\newcommand{\Cov}{{\mathbb{C}\rm{ov}}}
\newcommand{\Tr}{{\rm{Tr}}}

\newcommand{\bA}{\boldsymbol{A}}
\newcommand{\bB}{\boldsymbol{B}}

\newcommand{\bF}{\boldsymbol{F}}

\newcommand{\bH}{\boldsymbol{H}}

\newcommand{\bP}{\boldsymbol{P}}

\newcommand{\bK}{\boldsymbol{K}}
\newcommand{\bL}{\boldsymbol{L}}
\newcommand{\bI}{\boldsymbol{I}}

\newcommand{\A}{\mathcal{A}}
\newcommand{\Y}{\mathcal{Y}}

\newcommand{\N}{\mathcal{N}}
\newcommand{\bG}{\boldsymbol{G}}

\newcommand{\bGa}{\boldsymbol{\Gamma}}
\newcommand{\bPsi}{\boldsymbol{\Psi}}

\newcommand{\eg}{{\em e.g., }}

\newtheorem{theorem}{Theorem}
\newtheorem{lemma}{Lemma}

\begin{document}

\title{Partially Linear Estimation with Application to Sparse Signal Recovery From Measurement Pairs}

\author{Tomer Michaeli, Daniel Sigalov
and Yonina C. Eldar,~\IEEEmembership{Senior~Member,~IEEE} %
\thanks{This work was supported in part by the Israel Science Foundation under
Grant no. 1081/07, by the European Commission in the framework
of the FP7 Network of Excellence in Wireless COMmunications NEWCOM++
(contract no. 216715) and by a Google Research Award.%
} %
\thanks{T.~Michaeli and Y.~C.~Eldar are with the department of electrical engineering and D.~Sigalov is with the interdepartmental program for applied mathematics, Technion--Israel Institute of Technology, Haifa, Israel 32000 (phone: +972-4-8294798, +972-4-8294682, fax: +972-4-8295757, e-mail: tomermic@tx.technion.ac.il, dansigal@tx.technion.ac.il, yonina@ee.technion.ac.il). Y.~C.~Eldar is also a Visiting Professor at Stanford, CA USA.%
}}
\maketitle
\begin{abstract}
We address the problem of estimating a random vector $X$ from two sets of measurements $Y$ and $Z$, such that the estimator is linear in $Y$. We show that the partially linear minimum mean squared error (PLMMSE) estimator does not require knowing the joint distribution of $X$ and $Y$ in full, but rather only its second-order moments. This renders it of potential interest in various applications. We further show that the PLMMSE method is minimax-optimal among all estimators that solely depend on the second-order statistics of $X$ and $Y$. We demonstrate our approach in the context of recovering a signal, which is sparse in a unitary dictionary, from noisy observations of it and of a filtered version of it. We show that in this setting PLMMSE estimation has a clear computational advantage, while its performance is comparable to state-of-the-art algorithms. We apply our approach both in static and dynamic estimation applications. In the former category, we treat the problem of image enhancement from blurred/noisy image pairs, where we show that PLMMSE estimation performs only slightly worse than state-of-the art algorithms, while running an order of magnitude faster. In the dynamic setting, we provide a recursive implementation of the estimator and demonstrate its utility in the context of tracking maneuvering targets from position and acceleration measurements.
\end{abstract}


%%%%%%%%%%%%%%%%%%%%%%%%%%%%%%%%%%%%%%%%%%%%%%%%%%%%%%%%%%%%%%%%%%%%%%%%%%%%%%%%%%%%%%%%%%%%%%%%%%%%%%%%%%%%%%%%%%%%%%%%%%%%%%%%%%
%%%%%%%%%%%%%%%%%%%%%%%%%%%%%%%%%%%%%%%%%%%%%%%%%%%%%%%%%%%%%%%%%%%%%%%%%%%%%%%%%%%%%%%%%%%%%%%%%%%%%%%%%%%%%%%%%%%%%%%%%%%%%%%%%%
\section{Introduction}
\label{sec:introduction}
Bayesian estimation is concerned with the prediction of a random quantity $X$ based on a set of observations $Y$, which are statistically related to $X$. It is well known that the estimator minimizing the mean squared error (MSE) is given by the conditional expectation $\hat{X}=\EE[X|Y]$. There are various scenarios, however, in which the minimal MSE (MMSE) estimator cannot be used. This can either be due to implementation constraints, because of the fact that no closed form expression for $\EE[X|Y]$ exists, or due to lack of complete knowledge of the joint distribution of $X$ and $Y$. In these cases, one often resorts to linear estimation. The appeal of the linear MMSE (LMMSE) estimator is rooted in the fact that it possesses an easily implementable closed form expression, which merely requires knowledge of the joint first- and second-order moments of $X$ and $Y$.

For example, the amount of computation required for calculating the MMSE estimate of a jump-Markov Gaussian random process from its noisy version grows exponentially in time \cite{BB88}. By contrast, the LMMSE estimator in this setting possesses a simple recursive implementation, similar to the Kalman filter \cite{C94}. A similar problem arises in the area of sparse representations, in which the use of sparsity-inducing Gaussian mixture priors and of Laplacian priors is very common. The complexity of calculating the MMSE estimator under the former prior is exponential in the vector's dimension, calling for approximate solutions \cite{SPZ08}. The MMSE estimator under the latter prior does not possess a closed form expression \cite{G01}, which has motivated the use of alternative estimation strategies such as the maximum a-posteriori (MAP) method.

In practical situations, the reasons for not using the MMSE estimator may only apply to a subset of the measurements. In these cases, it may be desirable to construct an estimator that is linear in part of the measurements and nonlinear in the rest. Partially linear estimation was studied in the statistical literature in the context of regression \cite{HL07}. In this line of research, it is assumed that the conditional expectation $g(y,z)=\EE[X|Y=y,Z=z]$ is linear in $y$. The goal, then, is to approximate $g(x,y)$ from a set of examples $\{x_i,y_i,z_i\}$ drawn independently from the joint distribution of $X$, $Y$ and $Z$. In this paper, our goal is to derive the separable partially linear MMSE (PLMMSE) estimator. Namely, we do not make any assumptions on the structure of the MMSE estimate $\EE[X|Y,Z]$, but rather look for the estimator that minimizes the MSE among all functions $g(x,y)$ of the form $\bA y + b(z)$.

% For example, in multi-view regression problems, the goal is to construct an estimator of $X$ based on two sets of features $Y$ and $Z$ \cite{KF07}. In these applications, one may be given a large training set of examples $\{x_i,z_i\}$ drawn independently from the joint distribution $F_{XZ}(x,z)$ of $X$ and $Z$ and only a small number of examples $\{x_i,y_i\}$ drawn from $F_{XY}(x,y)$. Thus, $F_{XZ}(x,z)$ can be approximated from the first training set with great accuracy, for example using nonparametric techniques. However, due to its small cardinality, the second set can only be used to estimate the cross-covariance matrix $\bGa_{XY}$ of $X$ and $Y$, but not the entire distribution $F_{XY}(x,y)$. This implies that the MMSE estimator $\EE[X|Y,Z]$ cannot be computed and we must settle for estimators that do not require knowing $F_{XYZ}$. As we will see, one such approach is minimization of the MSE over the class of estimators that are linear in $Y$.

We show that in certain sparse approximation scenarios, the PLMMSE solution may be computed much more efficiently than the MMSE estimator. We demonstrate the usefulness of the sparse PLMMSE both in static and in dynamic estimation settings. In the static case, we apply our method to the problem of image deblurring from blurred/noisy image pairs \cite{YSQS07}. Here, we show that PLMMSE estimation performs only slightly worse than state-of-the art methods, but is much faster. In the dynamic regime, we provide a recursive implementation of the PLMMSE solution and demonstrate its usefulness in tracking a maneuvering target from position and acceleration measurements. We show the advantage of PLMMSE filtering over state-of-the-art algorithms when the measurements are prone to faults or contain outliers.

The paper is organized as follows. In Section~\ref{sec:PLE} we present the PLMMSE estimator and discuss some of its properties. In Section~\ref{sec:ParKno}, we show that the PLMMSE method is optimal in a minimax sense among all estimators that solely rely on the second-order statistics of $X$ and $Y$. In Section~\ref{sec:simulation} we derive the PLMMSE estimator for recovering a signal, sparse in a unitary dictionary, from a pair of observations, one blurred and one noisy. In Section~\ref{sec:simulationDeblur} we apply our method to the problem of image enhancement from blurred/noisy measurement pairs. In Section~\ref{sec:tracking} we apply PLMMSE estimation to tracking maneuvering targets.

%%%%%%%%%%%%%%%%%%%%%%%%%%%%%%%%%%%%%%%%%%%%%%%%%%%%%%%%%%%%%%%%%%%%%%%%%%%%%%%%%%%%%%%%%%%%%%%%%%%%%%%%%%%%%%%%%%%%%%%%%%%%%%%%%%
%%%%%%%%%%%%%%%%%%%%%%%%%%%%%%%%%%%%%%%%%%%%%%%%%%%%%%%%%%%%%%%%%%%%%%%%%%%%%%%%%%%%%%%%%%%%%%%%%%%%%%%%%%%%%%%%%%%%%%%%%%%%%%%%%%
\section{Partially Linear Estimation}
\label{sec:PLE}
We denote random variables (RVs) by capital letters. The pseudo-inverse of a matrix $\bA$ is denoted by $\bA^\dag$. The mean $\EE[X]$ of an RV $X$ is denoted $\mu_X$ and the auto-covariance matrix $\Cov(X)=\EE[(X-\mu_X)(X-\mu_X)^T]$ of $X$ is denoted $\bGa_{XX}$. Similarly, $\bGa_{XY}$ stands for the cross-covariance matrix $\Cov(X,Y)=\EE[(X-\mu_X)(Y-\mu_Y)^T]$ of two RVs $X$ and $Y$. The joint cumulative distribution function of $X$ and $Y$ is written as $F_{XY}(x,y)=\PP(X\leq x,Y\leq y)$, where the inequalities are element-wise. By definition, the marginal distribution of $X$ is $F_{X}(x)=F_{XY}(x,\infty)$. In our setting, $X$ is the quantity to be estimated and $Y$ and $Z$ are two sets of measurements thereof. The RVs $X$, $Y$ and $Z$ take values in $\RN^M$, $\RN^N$ and $\RN^Q$, respectively. The MSE of an estimator $\hat{X}$ of $X$ is defined as $\EE[\|X-\hat{X}\|^2]$.

We begin by considering the most general form of a partially linear estimator of $X$ based on $Y$ and $Z$, which is given by
\begin{equation}\label{eq:PLmodel1}
\hat{X} = \bA(Z) Y + b(Z).
\end{equation}
Here $\bA(z)$ is a matrix-valued function and $b(z)$ is a vector-valued function, so that the realization $z$ of $Z$ is used to choose one of a family of linear estimators of $x$ based on $y$.
\begin{theorem}\label{thm:PL}
Consider estimators of $X$ having the form \eqref{eq:PLmodel1}, for some (Borel measurable) functions $\bA:\RN^Q\rightarrow\RN^{M\times N}$ and $b:\RN^Q\rightarrow\RN^M$. Then the estimator minimizing the MSE within this class is given by
\begin{equation}\label{eq:XhatModel1}
\hat{X} = \bGa_{XY|Z}\bGa_{YY|Z}^\dag (Y - \EE[Y|Z]) + \EE[X|Z],
\end{equation}
where $\bGa_{XY|Z}=\EE[(X-\EE[X|Z])(Y-\EE[Y|Z])^T|Z]$ denotes the cross-covariance of $X$ and $Y$ given $Z$ and $\bGa_{YY|Z}=\EE[(Y-\EE[Y|Z])(Y-\EE[Y|Z])^T|Z]$ is the auto-covariance of $Y$ given $Z$.
\end{theorem}
\begin{proof}
See Appendix~\ref{sec:ProofThmPL}.
\end{proof}
Note that \eqref{eq:XhatModel1} is indeed of the form of \eqref{eq:PLmodel1} with $\bA(Z)=\bGa_{XY|Z}\bGa_{YY|Z}^\dag$ and $b(Z)=\EE[X|Z]-\bGa_{XY|Z}\bGa_{YY|Z}^\dag\EE[Y|Z]$. As can be seen, although the MMSE solution among the class of estimators \eqref{eq:PLmodel1} has a simple form, it requires knowing the conditional covariance $\bGa_{XY|Z}$, which limits its applicability. In particular, this solution cannot be applied in cases where we merely know the unconditional covariance $\bGa_{XY}$.

To relax this restriction, we next consider \emph{separable} partially linear estimation. Namely, we seek to minimize the MSE among all functions of the form
\begin{equation}\label{eq:PLmodel2}
\hat{X} = \bA Y + b(Z),
\end{equation}
where $\bA$ is a deterministic matrix and $b(z)$ is a vector-valued function.
\begin{theorem}
\label{thm:SPL}
Consider estimators of $X$ having the form \eqref{eq:PLmodel2}, for some matrix $\bA\in\RN^{M\times N}$ and (Borel measurable) function $b:\RN^Q\rightarrow\RN^M$. Then the estimator minimizing the MSE within this class is given by
\begin{equation}\label{eq:Xhat}
\hat{X}^{\rm PL} = \bGa_{X\tilde{Y}}\bGa_{\tilde{Y}\tilde{Y}}^\dag \tilde{Y} + \EE[X|Z],
\end{equation}
where
\begin{equation}\label{eq:W}
\tilde{Y} = Y - \EE[Y|Z].
\end{equation}
\end{theorem}
\begin{proof}
See Appendix~\ref{sec:ProofThmSPL}.
\end{proof}
Note again that \eqref{eq:Xhat} is of the form of \eqref{eq:PLmodel2} with $\bA=\bGa_{X\tilde{Y}}\bGa_{\tilde{Y}\tilde{Y}}^\dag$ and $b(Z)=\EE[X|Z]-\bGa_{X\tilde{Y}}\bGa_{\tilde{Y}\tilde{Y}}^\dag\EE[Y|Z]$. The major advantage of this solution with respect to the non-separable estimator \eqref{eq:PLmodel1}, is that the only required knowledge regarding the statistical relation between $X$ and $Y$ is of second-order type. Specifically, as we show in Appendix~\ref{sec:DerivationOfXhatExplicit}, \eqref{eq:Xhat} can be equivalently written as
\begin{align}\label{eq:XhatExplicit}
\hat{X}^{\rm PL} =& \left(\!\bGa_{XY}-\bGa_{\hat{X}^{\rm NL}_Z\hat{Y}^{\rm NL}_Z}\!\right)\!\left(\!\bGa_{YY}-\bGa_{\hat{Y}^{\rm NL}_Z\hat{Y}^{\rm NL}_Z}\!\right)^\dag \!\!\left(\!Y-\hat{Y}^{\rm NL}_Z\!\right) \nonumber\\
&+ \hat{X}^{\rm NL}_Z,
\end{align}
where we denoted $\hat{X}^{\rm NL}_Z=\EE[X|Z]$ and $\hat{Y}^{\rm NL}_Z=\EE[Y|Z]$. Therefore, all we need to know in order to be able to compute the separable PLMMSE estimator \eqref{eq:Xhat} is the covariance matrix $\bGa_{XY}$, the conditional expectation $\EE[X|Z]$ and the marginal joint cumulative distribution function $F_{YZ}$ of $Y$ and $Z$. This is illustrated in Fig.~\ref{fig:XYZ}. In fact, as we show in Section~\ref{sec:ParKno}, in addition to being optimal among all partially linear methods, the PLMMSE solution \eqref{eq:Xhat} is also optimal in a minimax sense among all estimation strategies that rely solely on the quantities appearing in Fig.~\ref{fig:XYZ}.

\begin{figure}[t]\centering
\includegraphics[scale=1]{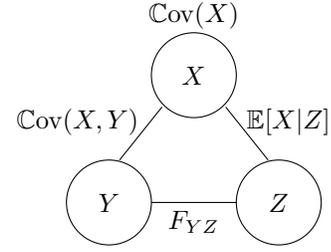}
\caption{The statistical knowledge required for computing the PLMMSE estimator \eqref{eq:Xhat}.}
\label{fig:XYZ}
\end{figure}

The intuition behind \eqref{eq:Xhat} is similar to that arising in dynamic estimation schemes, such as the Kalman filter. Specifically, we begin by constructing the estimate $\EE[X|Z]$ of $X$ based on the measurements $Z$, which minimizes the MSE among all functions of $Z$. Next, we would like to account for $Y$. However, since $Z$ has already been accounted for, we first need to subtract from $Y$ all variations caused by $Z$. This is done by constructing the RV $\tilde{Y}$ of \eqref{eq:W}, which can be thought of as the \emph{innovation} associated with the measurements $Y$ with respect to the initial estimate $\EE[X|Z]$. Finally, since we want an estimate that is partially linear in $Y$, we update our initial estimate with the LMMSE estimate of $X$ based on $\tilde{Y}$.

%%%%%%%%%%%%%%%%%%%%%%%%%%%%%%%%%%%%%%%%%%%%%%%%%%%%%%%%%%%%%%%%%%%%%%%%%%%%%%%%%%%%%%%%%%%%%%%%%%%%%%%%%%%%%%%%%%%%%%%%%%%%%%%%%%%
%\subsection{Discussion}
Before discussing the minimax-optimality of the PLMMSE estimator, it is insightful to examine several special cases, as we do next.

\paragraph{Independent measurements}
Consider first the case in which $Y$ and $Z$ are statistically independent. In this setting, $\tilde{Y}=Y-\mu_Y$ and therefore the PLMMSE estimator \eqref{eq:Xhat} becomes
\begin{align}\label{eq:IndMeas}
\hat{X}^{\rm PL} &= \bGa_{XY}\bGa_{YY}^\dag (Y - \mu_Y) + \EE[X|Z] = \hat{X}^{\rm L}_{Y}+\hat{X}^{\rm NL}_Z - \mu_X,
\end{align}
where $\hat{X}^{\rm L}_{Y}$ denotes the LMMSE estimate of $X$ from $Y$. Thus, in this setting, the PLMMSE estimate reduces to a linear combination of the LMMSE estimate $\hat{X}^{\rm L}_{Y}$ and the MMSE estimate $\hat{X}^{\rm NL}_Z$. The need for subtracting the mean of $X$ arises because both $\hat{X}^{\rm L}_{Y}$ and $\hat{X}^{\rm NL}_Z$ account for it. Indeed, note that $\EE[\hat{X}^{\rm L}_{Y}]=\EE[\hat{X}^{\rm NL}_Z]=\mu_X$, so that without subtraction of $\mu_X$, the estimate $\hat{X}^{\rm PL}$ would be biased, with a mean of $2\mu_X$.

%\paragraph{$X$ is a Function of $Z$}
%Suppose next that $X=f(Z)$ for some deterministic function $f(\cdot)$. Recall that $W$ is uncorrelated with every function of $Z$, and therefore it is in particular uncorrelated with $X$. Consequently, in this setting, $\bGa_{XW}=0$ and the PLMMSE estimator \eqref{eq:Xhat} reduces to
%\begin{equation}
%\hat{X}=\EE[X|Z]=f(Z),
%\end{equation}
%as expected.
%
%\paragraph{$X$ is a Linear Function of $Y$}
%Consider the situation in which $X=\bH Y$ for some deterministic matrix $\bH$ and assume that $\bGa_{WW}$ is strictly positive definite. Substituting $Y=W+\EE[Y|Z]$, we have in this case that
%\begin{align}
%\bGa_{XW} &= \bH\EE[(Y-\EE[Y])W^T] \nonumber\\
%          &= \bH\EE[(W+\EE[Y|Z]-\EE[Y])W^T] \nonumber\\
%          &= \bH\left(\bGa_{WW} + \EE[\EE[Y|Z]W^T] - \EE[Y]\EE[W^T]\right) \nonumber\\
%          &= \bH\bGa_{WW},
%\end{align}
%where we used the fact that $W$ has zero mean and is uncorrelated with $\EE[Y|Z]$. Consequently, the PLMMSE estimator \eqref{eq:Xhat} reduces to
%\begin{equation}
%\hat{X} = \bH\bGa_{WW}\bGa_{WW}^{-1} (Y-\EE[Y|Z]) + \bH\EE[Y|Z] = \bH Y.
%\end{equation}

\paragraph{$Z$ is independent of $X$ and $Y$}
Suppose next that both $X$ and $Y$ are statistically independent of $Z$. Thus, in addition to the fact that $\tilde{Y}=Y-\mu_Y$, we also have $\EE[X|Z]=\mu_X$. Consequently, the PLMMSE solution \eqref{eq:Xhat} reduces to the LMMSE estimate of $X$ given $Y$:
\begin{align}
\hat{X} &= \bGa_{XY}\bGa_{YY}^\dag (Y - \mu_Y) + \mu_X = \hat{X}^{\rm L}_{Y}.
\end{align}

\paragraph{$Y$ is uncorrelated with $X$ and independent of $Z$}
Consider the situation in which $Y$ and $Z$ are statistically independent and $X$ and $Y$ are uncorrelated. Then $\tilde{Y}=Y-\mu_Y$, and also $\bGa_{X\tilde{Y}}=\bGa_{XY}=0$ so that \eqref{eq:Xhat} becomes the MMSE estimate of $X$ from $Z$:
\begin{align}
\hat{X} &= \EE[X|Z] = \hat{X}^{\rm NL}_{Z}.
\end{align}

\paragraph{$X$ is independent of $Z$}
In situations where $X$ and $Z$ are statistically independent, one may be tempted to conclude that the PLMMSE estimator should not be a function of $Z$. However, this is not necessarily the case. Specifically, although the second term in \eqref{eq:Xhat} becomes the constant $\EE[X|Z]=\mu_X$ in this setting, it is easily verified that $\bGa_{X\tilde{Y}}=\bGa_{XY}$, so that the first term in \eqref{eq:Xhat} does not vanish unless $X$ is uncorrelated with $Y$. As a consequence, the PLMMSE estimator can be written as
\begin{equation}
\hat{X} = \bGa_{XY}\bGa_{\tilde{Y}\tilde{Y}}^\dag Y + \mu_X - \bGa_{XY}\bGa_{\tilde{Y}\tilde{Y}}^\dag \EE[Y|Z],
\end{equation}
in which the last term is a function of $Z$. This should come as no surprise, though, because if, for instance, $Y=X+Z$, then the optimal estimate is $\hat{X}=Y-Z$, even if $X$ and $Z$ are independent. This solution is clearly a function of $Z$.

\paragraph{$X$ is uncorrelated with $Y$}
A similar phenomenon occurs when $X$ and $Y$ are uncorrelated. Indeed in this case, $\bGa_{X\tilde{Y}}=-\bGa_{\hat{X}_Z^{\rm NL} \hat{Y}_Z^{\rm NL}}$, so that the first term in \eqref{eq:Xhat} does not vanish unless $\hat{X}_Z^{\rm NL}$ is uncorrelated with $\hat{Y}_Z^{\rm NL}$. Consequently, the estimator \eqref{eq:Xhat} can be expressed as
\begin{equation}
\hat{X}=-\bGa_{\hat{X}_Z^{\rm NL}\hat{Y}_Z^{\rm NL}}\bGa_{\tilde{Y}\tilde{Y}}^\dag Y + \bGa_{\hat{X}_Z^{\rm NL}\hat{Y}_Z^{\rm NL}}\bGa_{\tilde{Y}\tilde{Y}}^\dag \EE[Y|Z] + \EE[X|Z],
\end{equation}
in which the first term is clearly a linear function of $Y$.

\paragraph{Additive noise}
Perhaps the most widely studied measurement model corresponds to linear distortion and additive noise. Specifically, suppose that
\begin{equation}\label{eq:additiveNoise}
Y=\bH X+U, \quad Z=\bG X+V,
\end{equation}
where $\bH\in\RN^{N\times M}$ and $\bG\in\RN^{Q\times M}$ are given matrices and $U$ and $V$ are zero-mean RVs such that $X$, $U$ and $V$ are mutually independent. As we show in Section~\ref{sec:simulation}, there are situations in which the distribution of $X$ is such that the complexity of computing the MMSE estimator $\EE[X|Y,Z]$ is huge, whereas the complexity of computing $\EE[X|Z]$ is modest. In these cases one may prefer to resort to PLMMSE estimation. This method does not correspond to a convex combination of the LMMSE estimate of $X$ from $Y$ and the MMSE estimate of $X$ from $Z$, as might be suspected. Indeed, substituting $Y=\bH X+U$, we have that $\bGa_{XY}=\bGa_{XX}\bH^T$ and $\bGa_{YY}=\bH\bGa_{XX}\bH^T+\bGa_{UU}$. Furthermore, $\EE[Y|Z]=\bH\EE[X|Z]$, so that $\bGa_{\hat{X}_Z^{\rm NL}\hat{Y}_Z^{\rm NL}}=\bGa_{\hat{X}_Z^{\rm NL}\hat{X}_Z^{\rm NL}}\bH^T$ and $\bGa_{\hat{Y}_Z^{\rm NL}\hat{Y}_Z^{\rm NL}}=\bH\bGa_{\hat{X}_Z^{\rm NL}\hat{X}_Z^{\rm NL}}\bH^T$. Consequently, the PLMMSE estimator \eqref{eq:XhatExplicit} becomes
%\begin{align}
%\bGa_{XW} &= \EE[(X-\EE[X])(X+U-\EE[X|Z])^T] \nonumber\\
%          &= \EE[(X-\EE[X])(X-\EE[X]+\EE[X]-\EE[X|Z])^T] \nonumber\\
%          &= \bGa_{XX} - \bGa_{\hat{X}^Z\hat{X}^Z},
%\end{align}
%where we used the fact that $\EE[X\hat{X}^Z^T]=\EE[\hat{X}^Z\hat{X}^Z^T]$ due to the orthogonality principle. In a similar manner, we obtain
%\begin{align}
%\bGa_{WW} &= \EE[(X+U-\EE[X|Z])(X+U-\EE[X|Z])^T] \nonumber\\
%%          &= \EE[(X-\EE[X]+U+\EE[X]-\EE[X|Z])(X-\EE[X]+U+\EE[X]-\EE[X|Z])^T] \nonumber\\
%          &= \bGa_{XX} + \bGa_{UU} - \bGa_{\hat{X}^Z\hat{X}^Z}.
%\end{align}
%Therefore, the PLMMSE estimator \eqref{eq:Xhat} is given in this case by
\begin{equation}\label{eq:PLEadditiveNoise}
\hat{X} = \bA Y + (\bI - \bA\bH)\EE[X|Z],
\end{equation}
where $\bI$ is the identity matrix and $\bA$ is given by
\begin{align}\label{eq:PLEadditiveNoiseA}
\bA &= \left(\bGa_{XX} - \bGa_{\hat{X}_Z^{\rm NL}\hat{X}_Z^{\rm NL}}\right) \bH^T \nonumber\\
 &\hspace{1.5cm}\times \left(\bH\left(\bGa_{XX} - \bGa_{\hat{X}_Z^{\rm NL}\hat{X}_Z^{\rm NL}}\right)\bH^T + \bGa_{UU}\right)^\dag.
\end{align}
We see that, as opposed to a convex combination of $\hat{X}_Z^{\rm NL}$ and $\hat{X}^{\rm L}_{Y}$, the PLMMSE method reduces to a combination of $\hat{X}_Z^{\rm NL}$ and $Y$. Furthermore, the weights of this combination are matrices rather than scalars.
%In the next section, we demonstrate the usefulness of this technique.

As a toy example demonstrating this, suppose that $X$ is a scalar binary RV taking the values $\pm1$ with equal probability, that $H=G=1$, and that $U\sim\N(0,\sigma_{V}^2)$ and $V\sim\N(0,\sigma_{V}^2)$. It is easily verified that in this case
\begin{equation}
\hat{X}_Z^{\rm NL} = \EE[X|Z] = \frac{\N(Z-1;0,\sigma_V^2)-\N(Z+1;0,\sigma_V^2)}{\N(Z-1;0,\sigma_V^2)+\N(Z+1;0,\sigma_V^2)},
%\frac{\N(z-1;0,\sigma_{VV})-\N(z+1;0,\sigma_{VV})}{\N(z-1;0,\sigma_{VV})+\N(z+1;0,\sigma_{VV})}
\end{equation}
where $\N(\gamma;\mu,\sigma^2)$ denotes the Gaussian density function with mean $\mu$ and variance $\sigma^2$, evaluated at $\gamma$. Similarly,
\begin{equation}
\hat{X}^{\rm L}_Y = \frac{\sigma_{XY}}{\sigma_{Y}^2}(Y-\mu_Y)+\mu_X = \frac{1}{1+\sigma_{U}^2}Y,
%\frac{\N(z-1;0,\sigma_{VV})-\N(z+1;0,\sigma_{VV})}{\N(z-1;0,\sigma_{VV})+\N(z+1;0,\sigma_{VV})}
\end{equation}
where we used the facts that $\sigma_{Y}^2=\sigma_{X}^2+\sigma_{U}^2$ and $\sigma_{XY}=\sigma_{X}^2=1$. The PLMMSE estimator \eqref{eq:PLEadditiveNoise}, is therefore given by
\begin{equation}\label{eq:PLEexample}
\hat{X}^{\rm PL} = \gamma Y + (1-\gamma) \hat{X}_Z^{\rm NL},
\end{equation}
where $\gamma = (1-\sigma_{\hat{X}_Z^{\rm NL}}^2)/(1+\sigma_U^2-\sigma_{\hat{X}_Z^{\rm NL}}^2)$ (see \eqref{eq:PLEadditiveNoiseA}). Figure~\ref{fig:example} compares the MSE attained by the PLMMSE method to that of the naive convex-combination estimator
\begin{equation}\label{eq:NAIVEexample}
\hat{X}^{\rm naive} = \alpha \hat{X}^{\rm L}_Y + (1-\alpha) \hat{X}_Z^{\rm NL},
\end{equation}
for all $\alpha\in[0,1]$. As can be seen, when $\sigma_U=\sigma_V$, the MMSE of the PLMMSE method is roughly $12\%$ lower than the lowest MSE of the naive estimator. This advantage becomes less significant as $\sigma_U$ and $\sigma_V$ draw apart. As mentioned above, though, in multi-dimensional problems the PLMMSE method uses matrix weights rather than scalars, so that its potential for improvement over the naive estimator is yet greater.

\begin{figure*}[t]\centering
\subfloat[]{\includegraphics[width=0.3\textwidth, trim=1cm 0cm 1cm 0.5cm, clip]{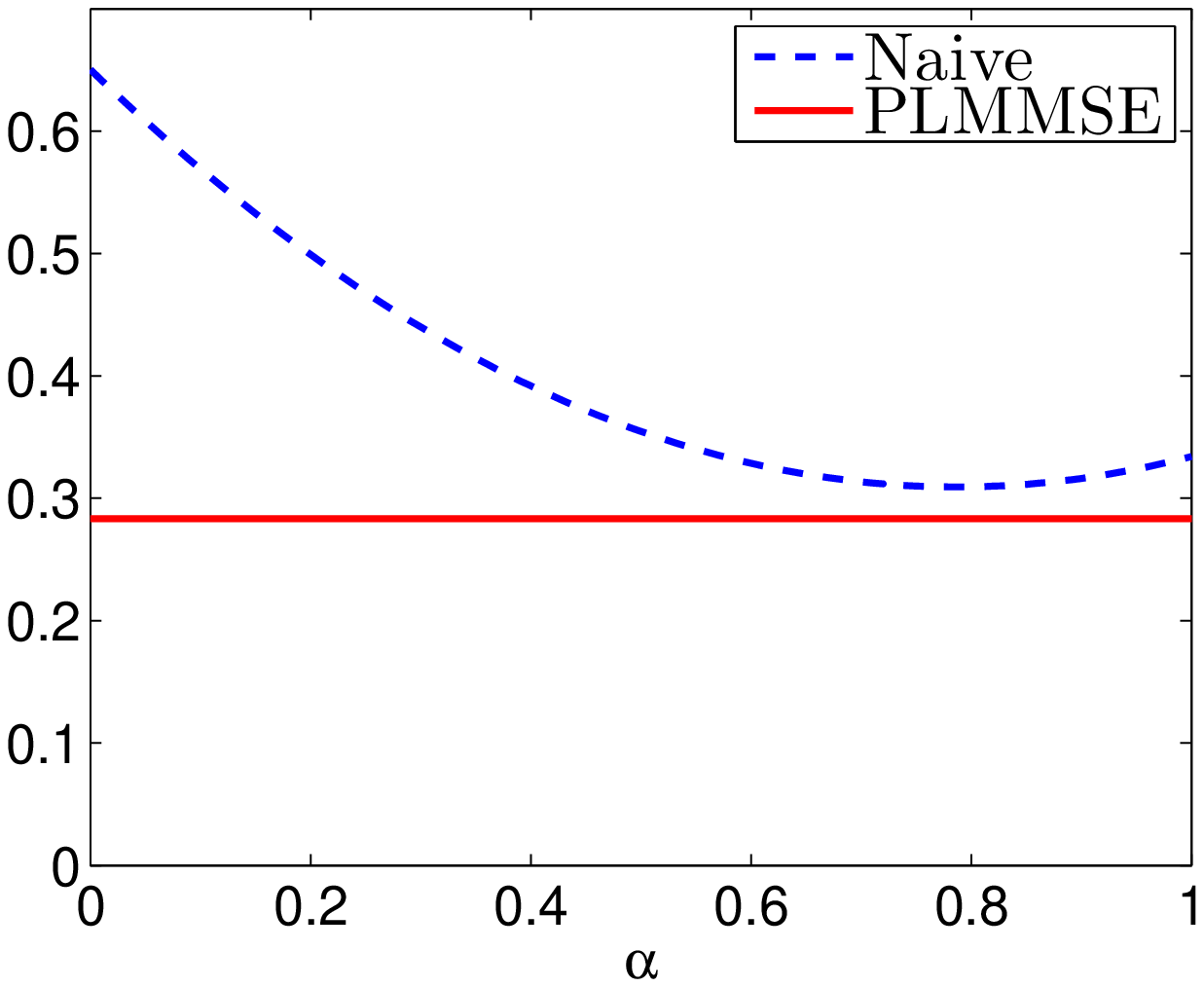}}
\subfloat[]{\includegraphics[width=0.3\textwidth, trim=1cm 0cm 1cm 0.5cm, clip]{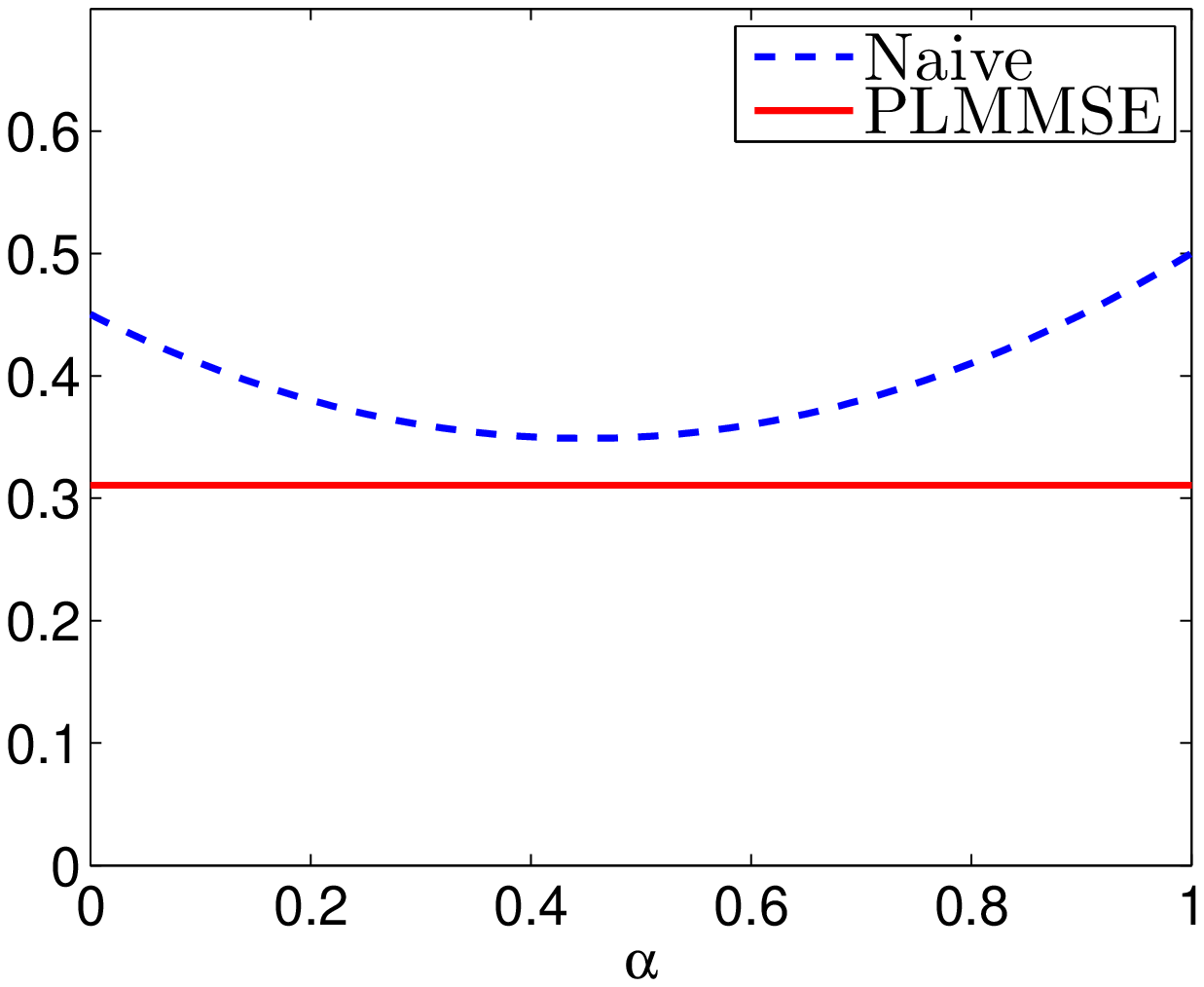}}
\subfloat[]{\includegraphics[width=0.3\textwidth, trim=1cm 0cm 1cm 0.5cm, clip]{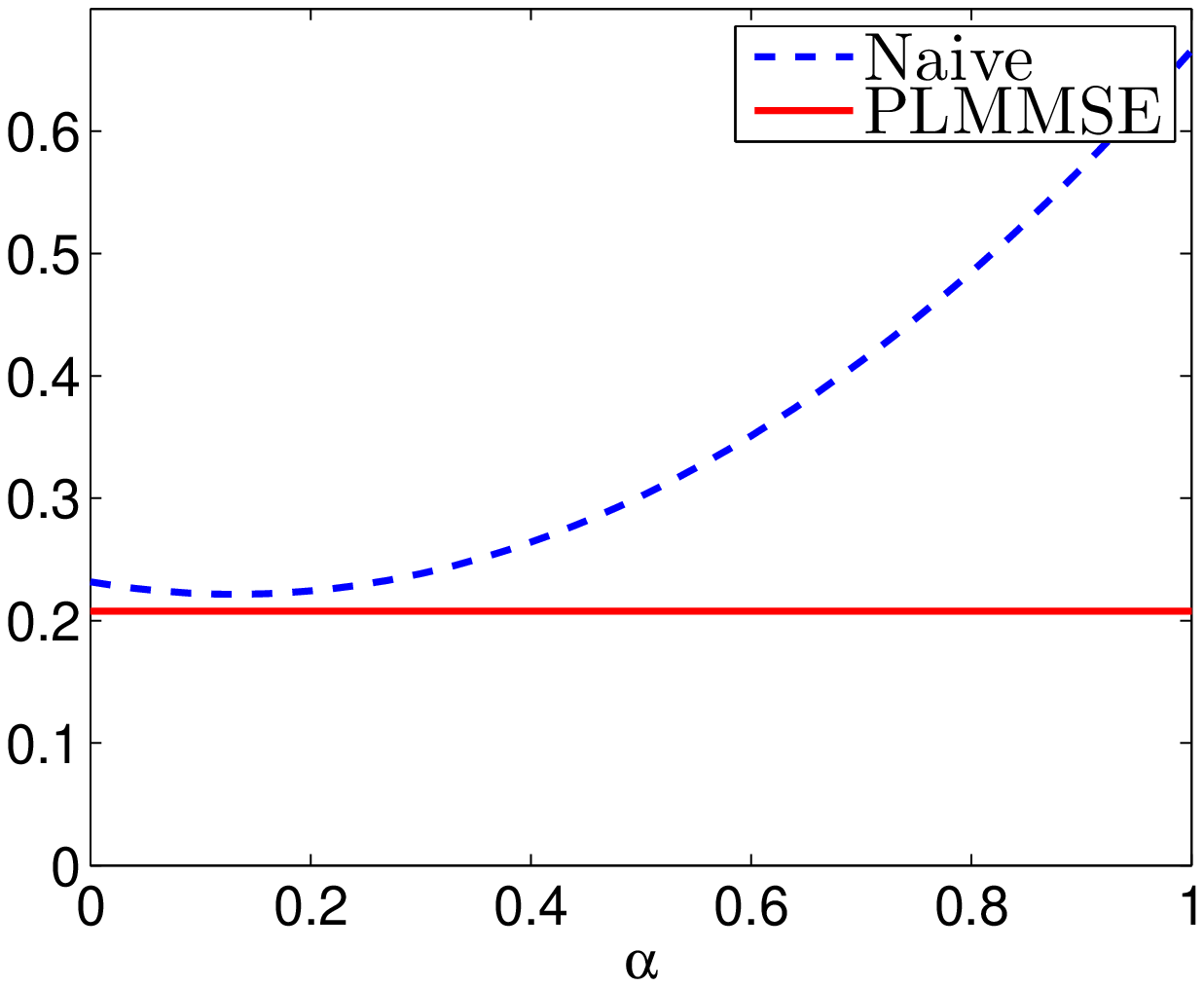}}
\caption{The MSE attained by $\hat{X}^{\rm PL}$ of \eqref{eq:PLEexample} and by $\hat{X}^{\rm naive}$ of \eqref{eq:NAIVEexample} as a function of $\alpha$ for several values of $\sigma_U$ and $\sigma_V$. (a)~$\sigma_U^2=0.5$, $\sigma_V^2=2$. (b)~$\sigma_U^2=1$, $\sigma_V^2=1$. (c)~$\sigma_U^2=2$, $\sigma_V^2=0.5$.}
\label{fig:example}
\end{figure*}

%%%%%%%%%%%%%%%%%%%%%%%%%%%%%%%%%%%%%%%%%%%%%%%%%%%%%%%%%%%%%%%%%%%%%%%%%%%%%%%%%%%%%%%%%%%%%%%%%%%%%%%%%%%%%%%%%%%%%%%%%%%%%%%%%%
%%%%%%%%%%%%%%%%%%%%%%%%%%%%%%%%%%%%%%%%%%%%%%%%%%%%%%%%%%%%%%%%%%%%%%%%%%%%%%%%%%%%%%%%%%%%%%%%%%%%%%%%%%%%%%%%%%%%%%%%%%%%%%%%%%
\section{Partial Knowledge of Statistical Relations}
\label{sec:ParKno}
As discussed in Section~\ref{sec:PLE}, one of the appealing properties of the PLMMSE solution is that it does not require knowing the entire joint distribution of $X$ and $Y$, but rather only its second-order moments. However, the fact that the PLMMSE estimator is merely determined by $\EE[X|Z]$, $\Cov(X,Y)$ and $F_{YZ}(y,z)$, does not yet imply that it is optimal among all methods that rely solely on these quantities. The question of optimality of an estimator with respect to partial knowledge regarding the joint distribution of the signal and measurements was recently addressed in \cite{ME11}. One of the notions of optimality considered there, which we adopt here as well, follows from a worst-case perspective. Specifically, any estimator $\hat{X}=g(Y,Z)$, may attain high MSE under certain distributions $F_{XYZ}(x,y,z)$ consistent with our knowledge and it may attain low MSE under other such distributions. We consider an estimator minimax-optimal if its worst-case MSE over the set of all feasible distributions is minimal. For example, it was shown in \cite{ME11} that the LMMSE estimator $\hat{X}_Y^{\rm L}$ attains the minimal possible worst-case MSE over the set of distributions $F_{XY}(x,y)$ with given first- and second-order moments.

In the next theorem we show that the PLMMSE method is optimal in the sense that its worst-case MSE over the set of all distributions $F_{XYZ}(x,y,z)$ complying with the knowledge appearing in Fig.~\ref{fig:XYZ} is minimal.
%\begin{align}\label{eq:setA}
%\A&=\left\{F_{XYZ}: \EE[X|Z]=g(Z), \, \EE[Y|Z]=h(Z), \right.\nonumber\\
%&\hspace{0.1cm} \Cov(g(Z),h(Z)) = \bGa_{\hat{X}^Z\hat{Y}^Z}, \Cov(h(Z)) = \bGa_{\hat{Y}^Z\hat{Y}^Z}\nonumber\\
%&\hspace{0.1cm}\left.\Cov(X)=\bGa_{XX},\, \Cov(X,Y)=\bGa_{XY}, \, \Cov(Y,Y)=\bGa_{YY}\right\}
%%&\hspace{1.25cm}\left.\Cov(Y,Y)=\bGa_{YY} \right\}
%\end{align}

\begin{theorem}\label{thm:PLMMSEminimax}
Let $\A$ be the set of probability distributions of $(X,Y,Z)$ satisfying
\begin{align}\label{eq:setA}
&\Cov(X)=\bGa_{XX}, \;\; \Cov(X,Y)=\bGa_{XY}, \;\;\EE[X|Z]=g(Z), \nonumber\\
&F_{XYZ}(\infty,y,z)=F_{YZ}(y,z),
%&\Cov(h(Z)) = \bGa_{\hat{Y}^Z\hat{Y}^Z}, \;\; \Cov(g(Z),h(Z)) = \bGa_{\hat{X}^Z\hat{Y}^Z}
%&F_{XYZ}(\infty,\infty,z)=F_Z(z)
\end{align}
where $\bGa_{XX}$ and $\bGa_{XY}$ are given matrices, $g(z)$ is a given function and $F_{YZ}(y,z)$ is a given cumulative distribution function. Then, among all estimators of $X$ based on $Y$ and $Z$, the PLMMSE method \eqref{eq:Xhat} has the minimal worst-case MSE
\begin{equation}\label{eq:minimaxLMMSE}
\sup_{F_{XYZ}\in\A} \EE_{F_{XYZ}}\!\left[\left\|X-\hat{X}\right\|^2\right],
\end{equation}
over the set $\A$.
\end{theorem}
\begin{proof}
See Appendix~\ref{sec:ProofOptPLMMSE}.
\end{proof}

\section{PLMMSE Estimation of Sparse Vectors}
\label{sec:simulation}
We now demonstrate the usefulness of the PLMMSE estimator in the context of sparse approximations. Specifically, consider the situation in which $X$ is known to be sparsely representable in a unitary dictionary $\bPsi\in\RN^{M\times M}$ in the sense that
\begin{equation}\label{eq:X_Psi_W}
X = \bPsi W
\end{equation}
for some RV $W$ that is sparse with high probability. More concretely, we assume that the elements of $W$ are given by
\begin{equation}\label{eq:RV_W}
W_i = S_i B_i,\quad i=1,\ldots M,
\end{equation}
where the RVs $\{B_i\}$ and $\{S_i\}$ are statistically independent, $B_i\sim\N(0,1)$ and $S_i=0$ (or take small values) with high probability. This assumption is very common in the sparse approximation literature. For example, in \cite{JXC08} and \cite{T01} the variables $S_i$ are assumed to follow a Gamma distribution. Here, we assume, as in \cite{SPZ08}, that they are binary, such that $\PP(S_i=\sigma_{1,i})=1-\PP(S_i=\sigma_{2,i})=p_i$ with some $\sigma_{1,i}\geq0$, $\sigma_{2,i}\geq0$ and $p_i\geq0$. In particular, setting $\sigma_{2,i}=0$ and $p_i$ small corresponds to vectors $W$ that are sparse with high probability.

Assume $X$ is observed through two linear systems, as in
\eqref{eq:additiveNoise}, where $\bH$ is an arbitrary matrix, $\bG$ is an orthogonal matrix satisfying $\bG^T\bG=\alpha^2\bI$ for some constant $\alpha\neq 0$, and $U$ and $V$ are Gaussian RVs with $\bGa_{UU}=\sigma_U^2\bI$ and $\bGa_{VV}=\sigma_V^2\bI$. A practical image enhancement scenario and a target tracking situation corresponding to this setting are detailed in sections~\ref{sec:simulationDeblur} and \ref{sec:tracking}, respectively. This setting can be cast in the standard sparse approximation form as
\begin{equation}\label{eq:measurements}
\begin{pmatrix}Y \\ Z\end{pmatrix}=\begin{pmatrix}\bH \\ \bG\end{pmatrix}X+\begin{pmatrix}U\\V\end{pmatrix}.
\end{equation}
It is well known that the expression for the MMSE estimate $\EE[X|Y,Z]$ in this case generally comprises $2^M$ summands, which correspond to all different possibilities of sparsity patterns in $W$ \cite{SPZ08}. This renders computation of the MMSE estimate prohibitively expensive even for modest values of $M$ and consequently various approaches have been devised to approximate this solution by a small number of terms (see \eg \cite{SPZ08} and references therein). For example, the fast Bayesian matching pursuit (FBMP) algorithm developed in \cite{SPZ08} employs a search in the tree representing all sparsity patterns in order to choose the terms participating in the approximation. We note that FBMP, as well as other sparse recovery methods, can operate with general measurement and dictionary matrices.

There are some special cases, however, in which the MMSE estimate possesses a simple structure, which can be implemented efficiently. As we show next, one such case is when both the channel's response and the dictionary over which $X$ is sparse correspond to orthogonal matrices. As in our setting $\bPsi$ is unitary and $\bG$ is orthogonal, this implies that we can efficiently compute the MMSE estimate $\EE[X|Z]$ of $X$ from $Z$. Therefore, instead of resorting to schemes for approximating $\EE[X|Y,Z]$, we can employ the PLMMSE estimator of $X$ based on $Y$ and $Z$, which possesses a closed form expression (see \eqref{eq:PLEadditiveNoise}) in this situation. This technique is particularly effective when the SNR of the observation $Y$ is much worse than that of $Z$, since the MMSE estimate $\EE[X|Y,Z]$ in this case is close to being partially linear in $Y$. Such a setting is demonstrated in Section~\ref{sec:simulationSub}.

%%%%%%%%%%%%%%%%%%%%%%%%%%%%%%%%%%%%%%%%%%%%%%%%%%%%%%%%%%%%%%%%%%%%%%%%%%%%%%%%%%%%%%%%%%%%%%%%%%%%%%%%%%%%%%%%%%%%%%%%%%%%%%%%%%
\subsection{MMSE Estimate of a Sparse Signal in a Unitary Dictionary}
\label{sec:sparse}

In our setting
\begin{equation}\label{eq:Zsparsity}
Z=\bG X + V = \bG \bPsi W + V,
\end{equation}
with $W$ of \eqref{eq:RV_W}. Since $\bG$ and $\bPsi$ are orthogonal, they are invertible, so that
\begin{equation}
\tilde{Z}=\frac{1}{\alpha}\bPsi^T\bG^T Z
\end{equation}
carries the same information on $X$ as $Z$ does, namely
\begin{equation}
\EE[X|Z] = \EE[X|\tilde{Z}] = \bPsi\EE[W|\tilde{Z}].
\end{equation}
Now, for every $i=1,\ldots,M$, we have that $\tilde{Z}_i = \alpha W_i + \tilde{V}_i$, where $\tilde{V}=\alpha^{-1}\bPsi^T\bG^T V$ is distributed $\N(0,\sigma_V^2\bI)$. Therefore, the set $\{\tilde{Z}_j\}_{j\neq i}$ is statistically independent of the pair $(W_i, \tilde{Z}_i)$ and consequently
\begin{align}\label{eq:EAgivenZ}
\EE[W_i|\tilde{Z}] &= \EE[W_i|\tilde{Z}_i] \nonumber\\
&= \EE[W_i|\tilde{Z}_i, S_i = \sigma_{1,i}]\PP(S_i = \sigma_{1,i}|\tilde{Z}_i) \nonumber\\
&\hspace{0.4cm}+ \EE[W_i|\tilde{Z}_i, S_i = \sigma_{2,i}]\PP(S_i = \sigma_{2,i}|\tilde{Z}_i).
\end{align}
Under the event $S_i=\sigma_{j,i}$ with a fixed $j\in\{1,2\}$, the RVs $W_i$ and $\tilde{Z}_i$ are jointly normally distributed with mean zero, implying that
\begin{equation}\label{eq:EAgivenS}
\EE[W_i|\tilde{Z}_i, S_i = \sigma_{j,i}] = \frac{\Cov(W_i,\tilde{Z}_i)}{\Cov(\tilde{Z}_i)} = \frac{\alpha\sigma_{j,i}^2}{\alpha^2\sigma_{j,i}^2+\sigma_V^2}\tilde{Z}_i.
\end{equation}
Finally, using Bayes rule, the term $\PP(S_i = \sigma_{1,i}|\tilde{Z}_i)$ reduces to
\begin{align}\label{eq:PSgivenZ}
&\frac{f_{\tilde{Z_i}|S_i}(\tilde{Z}_i|S_i = \sigma_{1,i})p_i}{f_{\tilde{Z_i}|S_i}(\tilde{Z}_i|S_i = \sigma_{1,i})p_i+f_{\tilde{Z_i}|S_i}(\tilde{Z}_i|S_i = \sigma_{2,i})(1-p_i)} \nonumber\\
&=\frac{\N(\tilde{Z}_i;0,\alpha^2\sigma_{1,i}^2+\sigma_V^2)p_i}{\N(\tilde{Z}_i;0,\alpha^2\sigma_{1,i}^2+\sigma_V^2)p_i+\N(\tilde{Z}_i;0,\sigma_{2,i}^2+\sigma_V^2)(1-p_i)}
\end{align}
and, similarly, $\PP(S_i = \sigma_{2,i}|\tilde{Z}_i)$ is given by
\begin{align}\label{eq:PSgivenZ}
\frac{\N(\tilde{Z}_i;0,\alpha^2\sigma_{2,i}^2+\sigma_V^2)(1-p_i)}{\N(\tilde{Z}_i;0,\alpha^2\sigma_{1,i}^2+\sigma_V^2)p_i+\N(\tilde{Z}_i;0,\sigma_{2,i}^2+\sigma_V^2)(1-p_i)}.
\end{align}
Substituting \eqref{eq:PSgivenZ} and \eqref{eq:EAgivenS} into \eqref{eq:EAgivenZ} leads to the following observation.
\begin{theorem}
The MMSE estimate of $X$ of \eqref{eq:X_Psi_W} given $Z$ of \eqref{eq:Zsparsity} is
\begin{equation}\label{eq:EXZsparse}
\EE[X|Z] = \bPsi \tilde{f}\left(\frac{1}{\alpha}\bPsi^T \bG^T Z\right),
\end{equation}
where $\tilde{f}(\tilde{z})=(f(\tilde{z}_1),\ldots,f(\tilde{z}_M))^T$, with
\begin{align}\label{eq:f}
&f(\tilde{z}_i) = \nonumber\\ &\frac{\tilde{z}_i \left(\frac{\alpha\sigma_{1,i}^2p_i\N(\tilde{z}_i;0,\alpha^2\sigma_{1,i}^2+\sigma_V^2)}{\alpha^2\sigma_{1,i}^2+\sigma_V^2} + \frac{(1 - p_i)\alpha\sigma_{2,i}^2\N(\tilde{z}_i;0,\alpha^2\sigma_{2,i}^2+\sigma_V^2)}{\alpha^2\sigma_{2,i}^2+\sigma_V^2}\right)} {p_i\,\N(\tilde{z}_i;0,\alpha^2\sigma_{1,i}^2+\sigma_V^2)+(1-p_i)\,\N(\tilde{z}_i;0,\alpha^2\sigma_{2,i}^2+\sigma_V^2)}.
\end{align}
\end{theorem}

Therefore, if, for example, $\bPsi$ is a wavelet basis and $\bG=\bI$ (so that $\alpha=1$), then $\EE[X|Z]$ can be efficiently computed by taking the wavelet transform of $Z$ (multiplication by $\bPsi^T$), applying a scalar shrinkage function on each of the coefficients (namely calculating $f(\tilde{z}_i)$ for the $i$th coefficient) and applying the inverse wavelet transform (multiplication by $\bPsi$) on the result. Note that the shrinkage curve \eqref{eq:f} is different than the soft-threshold operation, originally proposed in \cite{DJ95}. The latter can be obtained as the MAP solution with a Laplacian prior, whereas our function corresponds to the MMSE solution with a Gaussian mixture prior.

%%%%%%%%%%%%%%%%%%%%%%%%%%%%%%%%%%%%%%%%%%%%%%%%%%%%%%%%%%%%%%%%%%%%%%%%%%%%%%%%%%%%%%%%%%%%%%%%%%%%%%%%%%%%%%%%%%%%%%%%%%%%%%%%%%
\subsection{PLMMSE Estimate of a Sparse Signal}%Observations}%Measurements}
\label{sec:PLMMSEtwoObservations}
Equipped with a closed form expression for $\EE[X|Z]$, we can now obtain an expression for the PLMMSE estimator \eqref{eq:PLEadditiveNoise}. Specifically, we have that
\begin{equation}
\bGa_{XX} = \bPsi \bGa_{WW} \bPsi^T ,
\end{equation}
where $\bGa_{WW}$ is a diagonal matrix with $(\bGa_{WW})_{i,i}=p_i\sigma_{1,i}^2+(1-p_i)\sigma_{2,i}^2$. Similarly,
\begin{equation}
\bGa_{\hat{X}_Z^{\rm NL}\hat{X}_Z^{\rm NL}} = \bPsi \Cov(\tilde{f}(\tilde{Z})) \bPsi^T ,
\end{equation}
where $\Cov(\tilde{f}(\tilde{Z}))$ is a diagonal matrix whose $(i,i)$ element is $\beta_i=\Cov(f(\tilde{Z_i}))$. This is due to the fact that the elements of $\tilde{Z}$ are statistically independent and the fact that the function $\tilde{f}(\cdot)$ operates element-wise on its argument. Therefore, the PLMMSE estimator is given in our setting by equation \eqref{eq:PLEadditiveNoise} with $\EE[X|Z]$ of \eqref{eq:EXZsparse} and with the matrix
\begin{align}\label{eq:PLEadditiveNoiseA_app}
\bA &= \bPsi\left(\bGa_{WW} - \Cov(\tilde{f}(\tilde{Z}))\right) \bPsi^T \bH^T \nonumber\\
&\hspace{0.4cm}\times\left(\bH\bPsi\left(\bGa_{WW} - \Cov(\tilde{f}(\tilde{Z}))\right) \bPsi^T\bH^T + \sigma_U^2\bI\right)^\dag.
\end{align}
Observe that there is generally no closed form expression for the scalars $\beta_i=\Cov(f(\tilde{Z_i}))$, rendering it necessary to compute them numerically. Since each $\beta_i$ is the variance of a scalar RV, it can be computed very efficiently, either by approximating the corresponding integral by a sum over a set of points on the real line or by Monte Carlo simulation. In Section~\ref{sec:simulationDeblur} we demonstrate how this can be done in a practical scenario.

An important special case corresponds to the setting in which $p_i=p$, $\sigma^2_{1,i}=\sigma^2_{1}$, and $\sigma^2_{2,i}=\sigma^2_{2}$ for every $i$. In this situation, we also have that $\beta_i=\beta$ for every $i$. Furthermore,
\begin{equation}
\bGa_{WW}= (p\sigma_1^2+(1-p)\sigma_2^2) \bI
\end{equation}
%and
%\begin{equation}
%\bGa_{\hat{X}^Z\hat{X}^Z} = \bPsi (\beta \bI) \bPsi^T = \beta \bI,
%\end{equation}
so that $\bA$ is simplified to
\begin{equation}\label{eq:PLEadditiveNoiseAsimple}
\bA = \bH^T \left(\bH\bH^T + \frac{\sigma_U^2}{p\sigma_1^2+(1-p)\sigma_2^2-\beta}\bI\right)^\dag.
\end{equation}
As can be seen, in this setting $\bA$ does not involve multiplication by $\bPsi$ or $\bPsi^T$. Thus, if $\bH$ corresponds to a convolution operation, then $\bA$ also corresponds to a filter, which can be efficiently applied in the Fourier domain.

%%%%%%%%%%%%%%%%%%%%%%%%%%%%%%%%%%%%%%%%%%%%%%%%%%%%%%%%%%%%%%%%%%%%%%%%%%%%%%%%%%%%%%%%%%%%%%%%%%%%%%%%%%%%%%%%%%%%%%%%%%%%%%%%%%
\subsection{Numerical Study}
\label{sec:simulationSub}

We now compare via simulations the MSE attained by $\hat{X}^{\rm PL}$ to that attained by $\hat{X}_{Z}^{\rm NL}$, $\hat{X}_{Y}^{\rm L}$ and the approximation to $\EE[X|Y,Z]$ produced by the FBMP method. Since we generate the signal $X$ and measurements $Y$ and $Z$ according to the assumed model, we do not compare our method to other Baysian approaches, such as Bayesian compressive sensing (BSC) \cite{JXC08} and sparse Bayesian learning (SBL) \cite{T01}, which assume a different generative model. Nevertheless, we note that a practical scenario, which deviates from the assumptions of all these methods, was studied in \cite{SPZ08} and showed that the performance of FBMP is commonly better than that of BSC and SBL. In terms of running time, FBMP is typically an order of magnitude faster than SBL and roughly twice as slow as BSC.

In our experiment $\bPsi \in \RN^{256\times 256}$ was taken to be a Hadamard matrix with normalized columns. The matrix $\bH$ corresponded to (circular) convolution with the sequence $h[n]=\exp\{-|n|/8.5\}$ and $\bG$ was taken to be diagonal. To comply with the assumption made in \cite{SPZ08} that the columns of the measurement matrix are normalized, we normalized the columns of $\bH$ to be of norm $0.99$ and set $\bG=0.01\bI$. We set $p_i=p$, $\sigma^2_{1,i}=\sigma^2_{1}$, and $\sigma^2_{2,i}=0$ for every $i$, so that $X$ was truly sparse with high probability. Figure~\ref{fig:example} depicts the MSE of all estimators as a function of the input SNR, which we define as $10\log_{10}(p\sigma_1^2/\sigma^2)$. As can be seen, the MSE of the PLMMSE method is significantly lower then that of $\hat{X}_{Z}^{\rm NL}$ and $\hat{X}_{Y}^{\rm L}$ and is very close to that attained by the FBMP method. At low SNR levels and low sparsity levels (high $p$) the performance of the PLMMSE method is even slightly better than the FBMP approach.

The average running time of the PLMMSE method was $0.6$msec for all tested values of $p$. The average running times of the FBMP method were $52.7$msec, $79.6$msec and $125.2$msec, respectively, for $p=1/3$, $p=1/2$ and $p=2/3$. The ratio between the computational cost of the two approaches, which was two orders of magnitude in this experiment, becomes higher as the dimension of $X$ is increased. At certain dimensions, such as images of size $512\times 512$ (in which case $M=512^2$), the FBMP method becomes impractical to apply while PLMMSE estimation can still be used very efficiently.

A word of caution is in place, though. In situations in which the SNR of the measurement $Y$ is roughly the same as that of $Z$ (or better), the FBMP method is advantageous in terms of performance. Therefore in this regime, decision on the use of PLMMSE estimation boils down a performance-complexity tradeoff.

\begin{figure*}[t]\centering
\subfloat[]{\includegraphics[width=0.33\textwidth, trim=1cm 0.2cm 1cm 0.5cm, clip]{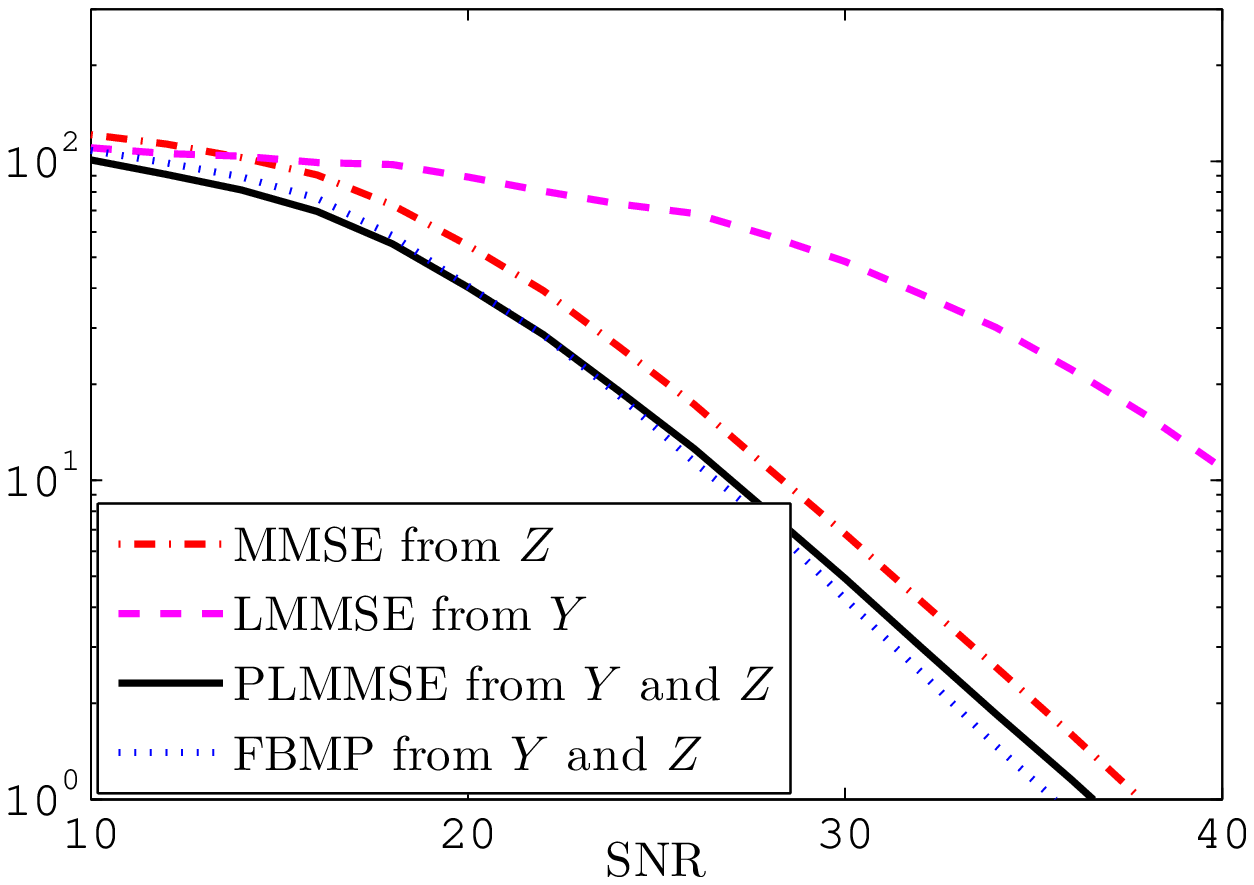}}%\hspace{0.2cm}
\subfloat[]{\includegraphics[width=0.33\textwidth, trim=1cm 0.2cm 1cm 0.5cm, clip]{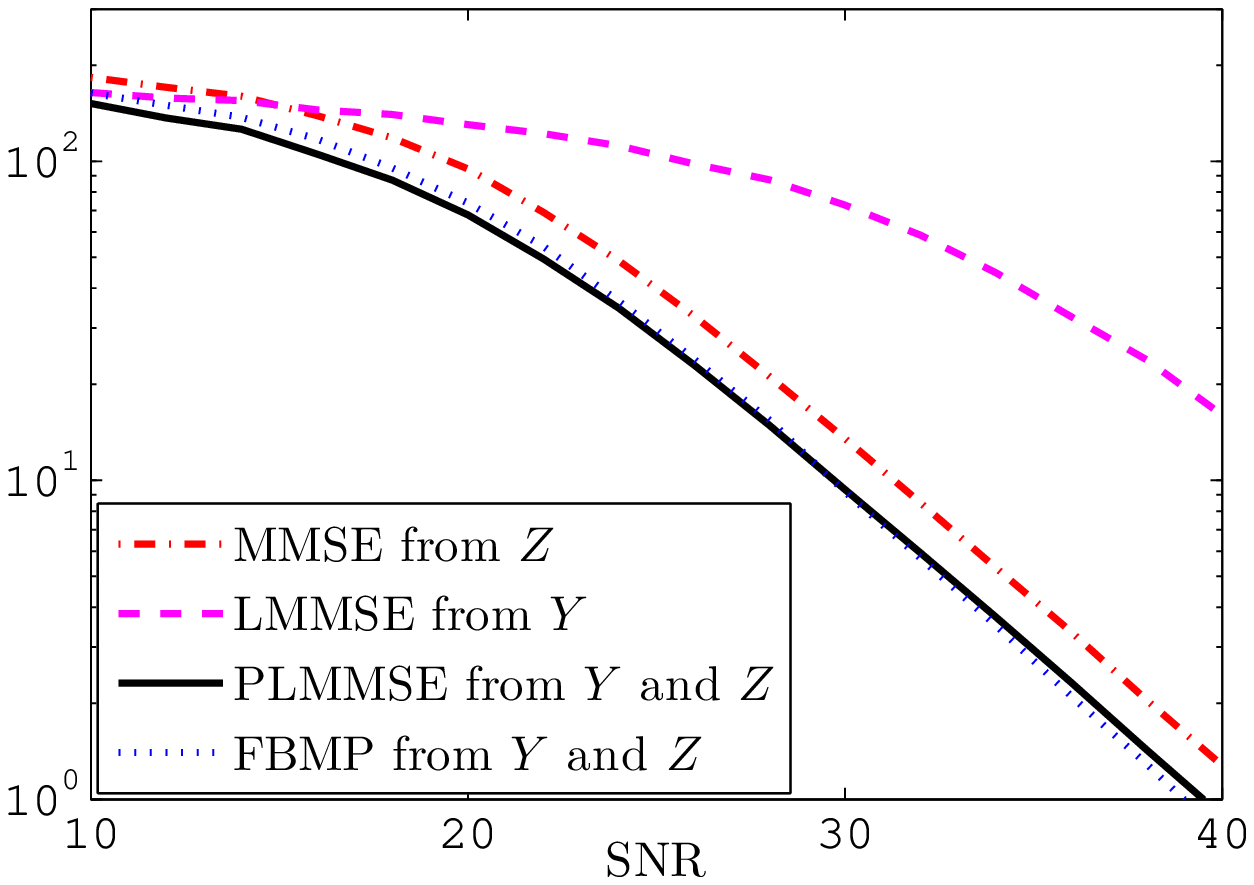}}%\hspace{0.2cm}
\subfloat[]{\includegraphics[width=0.33\textwidth, trim=1cm 0.2cm 1cm 0.5cm, clip]{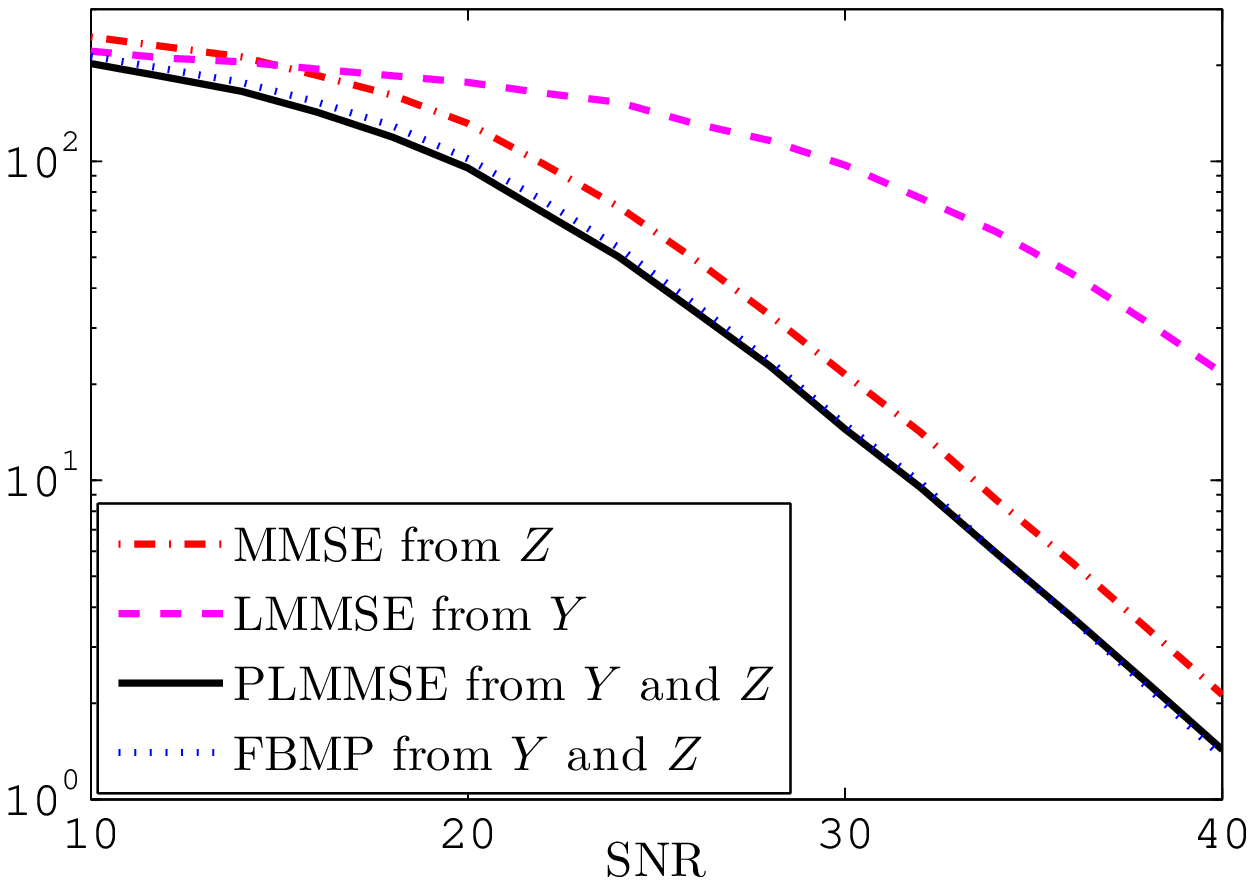}}
\caption{The MSE attained by $\hat{X}_{Z}^{\rm NL}$, $\hat{X}^{\rm L}_{Y}$, $\hat{X}^{\rm PL}$ and the approximation of $\EE[X|Y,Z]$ produced by the FBMP method \cite{SPZ08}. (a)~$p=1/3$. (b)~$p=1/2$. (c)~$p=2/3$.}
\label{fig:example}
\end{figure*}

%%%%%%%%%%%%%%%%%%%%%%%%%%%%%%%%%%%%%%%%%%%%%%%%%%%%%%%%%%%%%%%%%%%%%%%%%%%%%%%%%%%%%%%%%%%%%%%%%%%%%%%%%%%%%%%%%%%%%%%%%%%%%%%%%%
%%%%%%%%%%%%%%%%%%%%%%%%%%%%%%%%%%%%%%%%%%%%%%%%%%%%%%%%%%%%%%%%%%%%%%%%%%%%%%%%%%%%%%%%%%%%%%%%%%%%%%%%%%%%%%%%%%%%%%%%%%%%%%%%%%
\section{Application to Image Deblurring with Blurred/Noisy Image Pairs}
\label{sec:simulationDeblur}

When taking photos in dim light using a hand-held camera, there is a tradeoff between noise and motion blur, which can be controlled by tuning the shutter speed. Indeed, when using a long exposure time, the image typically comes out blurred due to camera shake. On the other hand, with a short exposure time (and high camera gain), the image is very noisy. In \cite{YSQS07} it was demonstrated how a high quality image can be constructed by properly processing two images of the same scene, one blurred and one noisy.

We now show how the PLMMSE approach can be applied in this setting to obtain plausible recoveries at a speed several orders of magnitude faster than any other sparsity-based method. In our setting $X$, $Y$ and $Z$ correspond, respectively, to the original, blurred (and slightly noisy) and noisy images. Thus, the measurement model is that described by \eqref{eq:measurements}, where $\bH$ corresponds to spatial convolution with some blur kernel, $\bG=\bI$, and $U$ and $V$ correspond to white Gaussian images with small and large variances respectively. We further assume that the image $X$ is sparse in some orthogonal wavelet basis $\bPsi$, such that it can be written as in \eqref{eq:X_Psi_W} and \eqref{eq:RV_W}.

% with $p_i=p$ and $\sigma_{B_i}^2=\sigma_B^2$ for every $i$

As we have seen, in this setting, the PLMMSE estimator can be computed in two stages. In the first stage, we calculate $\hat{X}_Z^{\rm NL}=\EE[X|Z]$ (namely, denoise the image $Z$) by computing the wavelet transform $\tilde{Z}=\bPsi^T Z$, applying the scalar shrinkage function \eqref{eq:f} on each wavelet coefficient,
%\begin{equation}
%f(\tilde{z}_i) = \frac{\frac{\alpha\sigma_{B_i}^2}{\alpha^2\sigma_{B_i}^2+\sigma_V^2}\,p\,\N(\tilde{z}_i;0,\alpha^2\sigma_{B_i}^2+\sigma_V^2)\,\tilde{z_i}}{p\,\N(\tilde{z_i};0,\alpha^2\sigma_{B_i}^2+\sigma_V^2)+(1-p)\,\N(\tilde{z_i};0,\sigma_V^2)}
%\end{equation}
%on the $i$th wavelet coefficient,
and taking the inverse wavelet transform of the result. This stage requires knowledge of the parameters $\{p_i\}$, $\{\sigma_{1,i}\}$, $\{\sigma_{2,i}\}$ and $\sigma_V$. To this end, we assume that $\sigma_{2,i}=0$ for all $i$ (a truly sparse image) and that $p_i$ and $\sigma_{1,i}$ are the same for wavelets coefficients at the same level. In other words, all wavelet coefficients of the noisy image $Z$ at level $\ell$ correspond to independent draws from the Gaussian mixture
\begin{equation}
f_{\tilde{Z_i}}(\tilde{z}) = p^\ell\N(\tilde{z};0,\alpha^2\sigma_{1,\ell}^2+\sigma_V^2)+(1-p^\ell)\N(\tilde{z};0,\sigma_V^2).
\end{equation}
Consequently, $p^\ell$, $\sigma_{1,\ell}$ and $\sigma_V$ can be estimated by expectation maximization (EM). In our experiments, we assumed that $\sigma_V$ is known and thus did not estimate it.

In the second stage, the denoised image $\hat{X}_Z^{\rm NL}$ needs to be combined with the blurred image $Y$ using \eqref{eq:PLEadditiveNoise} with $\bA$ of \eqref{eq:PLEadditiveNoiseA_app}. As discussed in Section~\ref{sec:PLMMSEtwoObservations}, this can be carried out very efficiently if $p_i=p$ and $\sigma_{1,i}=\sigma_1$ for all $i$. For the sake of efficiency\footnote{The exact solution involving \eqref{eq:PLEadditiveNoiseA_app} can be computed by using iterative techniques for matrix inversion, in which each iteration comprises filtering operations and forward and inverse wavelet transforms. However, we found that in most cases this approach leads to improvement of only $0.2$dB-$0.6$dB in PSNR and is much more demanding computationally.}, we therefore abandon the assumption that $p_i$ and $\sigma_{1,i}$ vary across wavelet levels and assume henceforth that all wavelet coefficients are independent and identically distributed. In this case, $\bA$ corresponds to the filter
\begin{equation}
A(\omega) = \frac{(\sigma_W^2 - \beta) H^*(\omega)}{(\sigma_W^2-\beta)|H(\omega)|^2 + \sigma_U^2},
\end{equation}
where $H(\omega)$ is the frequency response of the blur kernel. Consequently, the final PLMMSE estimate corresponds to the inverse Fourier transform of
\begin{equation}
\hat{X}_{\rm PLMMSE}^{\rm F}(\omega) = \frac{(\sigma_W^2 - \beta) H^*(\omega) Y^{\rm F}(\omega) + \sigma_{U}^2\hat{X}^{\rm F}_Z(\omega)} {(\sigma_W^2-\beta)|H(\omega)|^2 + \sigma_U^2},
\end{equation}
where $Y^{\rm F}(\omega)$ and $\hat{X}^{\rm F}_Z(\omega)$ denote the Fourier transforms of $Y$ and $\hat{X}_Z^{\rm NL}$, respectively. In our experiment, we assumed that the blur $H(\omega)$ and noise variance $\sigma_U^2$ are known. In practice, they can be estimated from $Y$ and $Z$, as proposed in \cite{YSQS07}. This stage also requires knowing the scalars $\sigma^2_W=\EE[W^2]$ and $\beta=\EE[f^2(\tilde{z})]$, which we estimate as
\begin{align}\label{eq:approxBeta}
\widehat{\sigma_W^2}=\frac{1}{M}\sum_{i=1}^M \tilde{z}_i^2 - \sigma_V^2,\quad \widehat{\beta}=\frac{1}{M}\sum_{i=1}^M f^2(\tilde{z}_i).
%\int_{\infty}^\infty f^2(\tilde{z}) f_{\tilde{Z_i}}(\tilde{z}) d\tilde{z} \approx \Delta \sum_{k=-K}^K f^2(k\Delta) f_{\tilde{Z_i}}(k\Delta)
\end{align}

Figure~\ref{fig:deblur} demonstrates our approach on the $512\times 512$ Gold-hill image. In this experiment, the blur corresponded to a Gaussian kernel with standard deviation $3.2$.
%\begin{equation}
%h[m,n] = \frac{1}{\sqrt{2\pi\sigma_{\rm PSF}^2}}\exp\left\{-\frac{2n^2+m^2}{2\sigma_{\rm PSF}^2}\right\},
%\end{equation}
%with width $\sigma_{\rm PSF}=3.2$ and gain $\delta=1.2$.
To model a situation in which the noise in $Y$ is due only to quantization errors, we chose $\sigma_U=1/\sqrt{12}\approx 0.3$ and $\sigma_V=45$. These parameters correspond to a peak signal to noise ratio (PSNR) of $25.08$dB for the blurred image and $15.07$dB for the noisy image.

We used the orthogonal Symlet wavelet of order $4$ and employed $10$ EM iterations to estimate $p^\ell$ and $\sigma_{1,\ell}^2$ in each wavelet level. The entire process takes $1.1$ seconds on a Dual-Core $3$GHz computer with un-optimized Matlab code. We note that our approach can be viewed as a smart combination of Wiener filtering for image debluring and wavelet thresholding for image denoising, which are among the simplest and fastest methods available. Consequently, the running time is at least an order of magnitude faster than any other sparsity-based method, including the Bayesian approaches FBMP \cite{SPZ08}, BCS \cite{JXC08} and SBL \cite{T01} and fast $\ell_1$ minimization algorithms such as NESTA \cite{BBC09}, GPSR \cite{FNW07} and Bergman iterations \cite{YOGD08}. As an example, the authors of \cite{SPZ08} reported that FBMP requires a runtime of $38$ minutes to recover a $128\times 128$ image from a few thousands of measurements and that GPSR requires $2.7$ minutes for the same task. BCS \cite{JXC08} was reported to require $15$ seconds for reconstructing a $512\times 512$ image from a few thousands of samples.

As can be seen in Figure~\ref{fig:deblur}, the quality of the recoveries corresponding to the denoised image $\hat{X}_Z^{\rm NL}$ and deblurred image $\hat{X}_Y^{\rm L}$ is rather poor with respect to the state-of-the-art BM3D debnoising method \cite{DFKE07} and BM3D debluring algorithm \cite{DFKE08}. Nevertheless, the quality of the joint estimate $\hat{X}_{\rm PLMMSE}$ surpasses each of these techniques alone. The residual deconvolution (RD) algorithm\footnote{In our setting this method does not produce ringing effects and thus the additional de-ringing stage proposed in \cite{YSQS07} was not applied.} proposed in \cite{YSQS07} for joint debluring and denoising slightly outperforms the PLMMSE method in terms of recovery error.

A quantitative comparison on several test images is provided in Table~\ref{tbl:comparison}. This comparison shows that the PSNR attained by the PLMMSE method is, on average, $0.3$dB higher than BM3D debluring, $0.4$dB higher than BM3D denoising, and $0.8$dB lower than RD. In terms of running times, however, our method is, on average, $11$ times faster than BM3D deblurring, $16$ times faster than BM3D denoising and $18$ times faster than RD. Note that RD requires initialization with a denoised version of $Z$, for which purpose we used the BM3D algorithm. Consequently, the running time reported in the last column of Table~\ref{tbl:comparison} includes the running time of the BM3D denoising method.

\begin{figure*}[t]\centering
\subfloat[]{\includegraphics[width=0.33\textwidth, trim=2cm 1.5cm 2cm 0.5cm, clip]{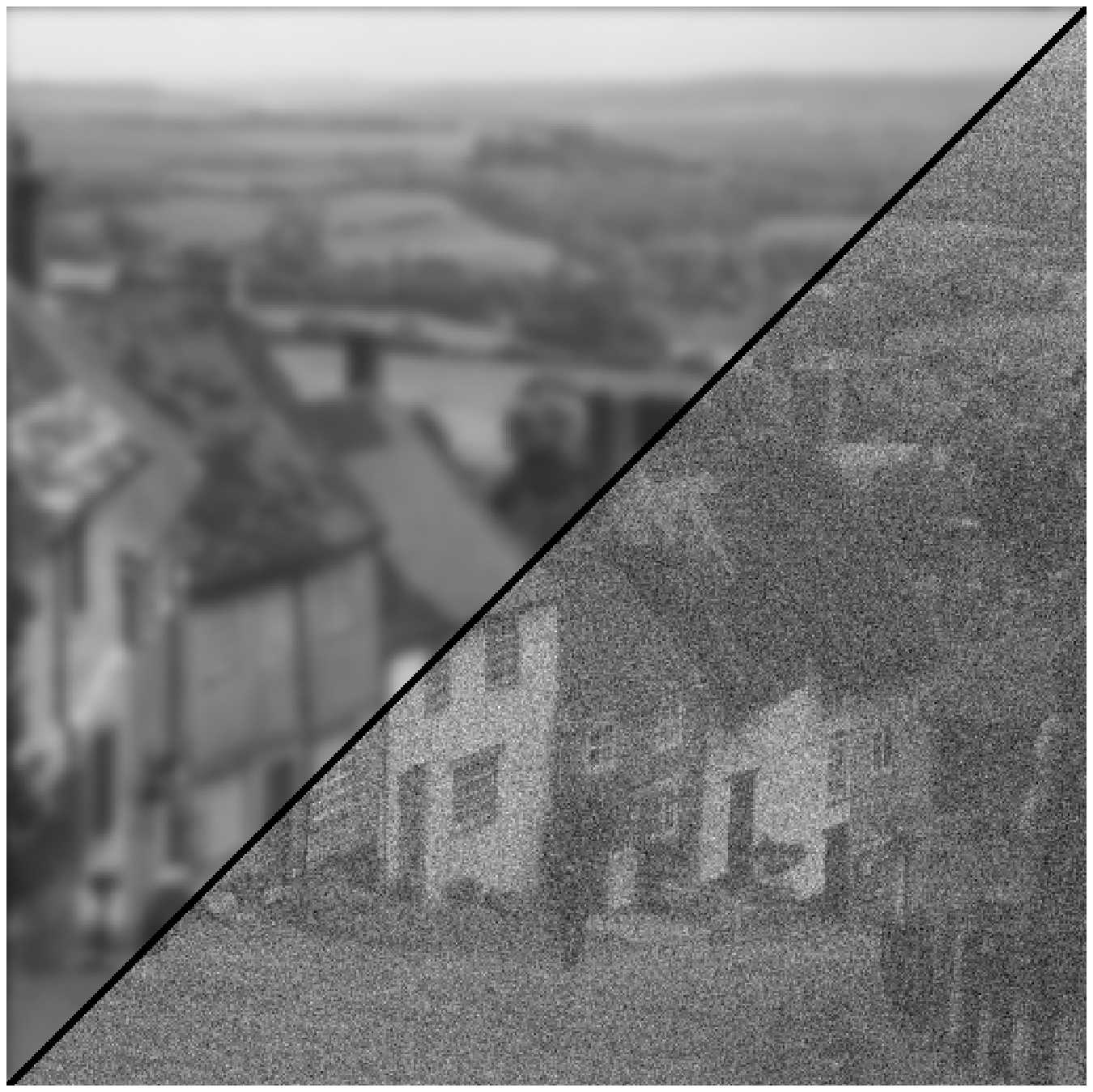} \label{fig:YZ}}%\hspace{0.2cm}
\subfloat[]{\includegraphics[width=0.33\textwidth, trim=2cm 1.5cm 2cm 0.5cm, clip]{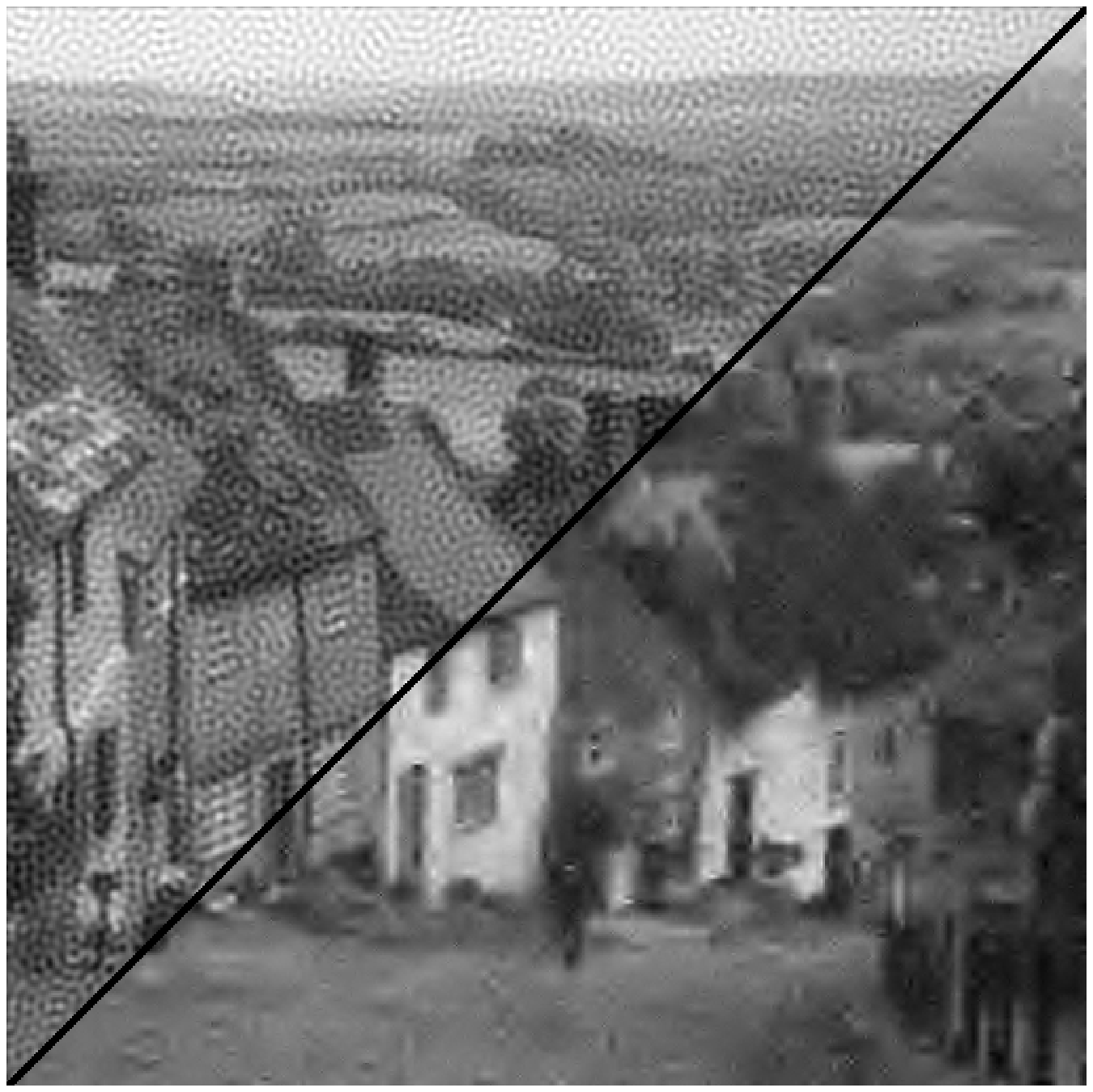} \label{fig:LMMSEy_MMSEz}}%\hspace{0.2cm}
\subfloat[]{\includegraphics[width=0.33\textwidth, trim=2cm 1.5cm 2cm 0.5cm, clip]{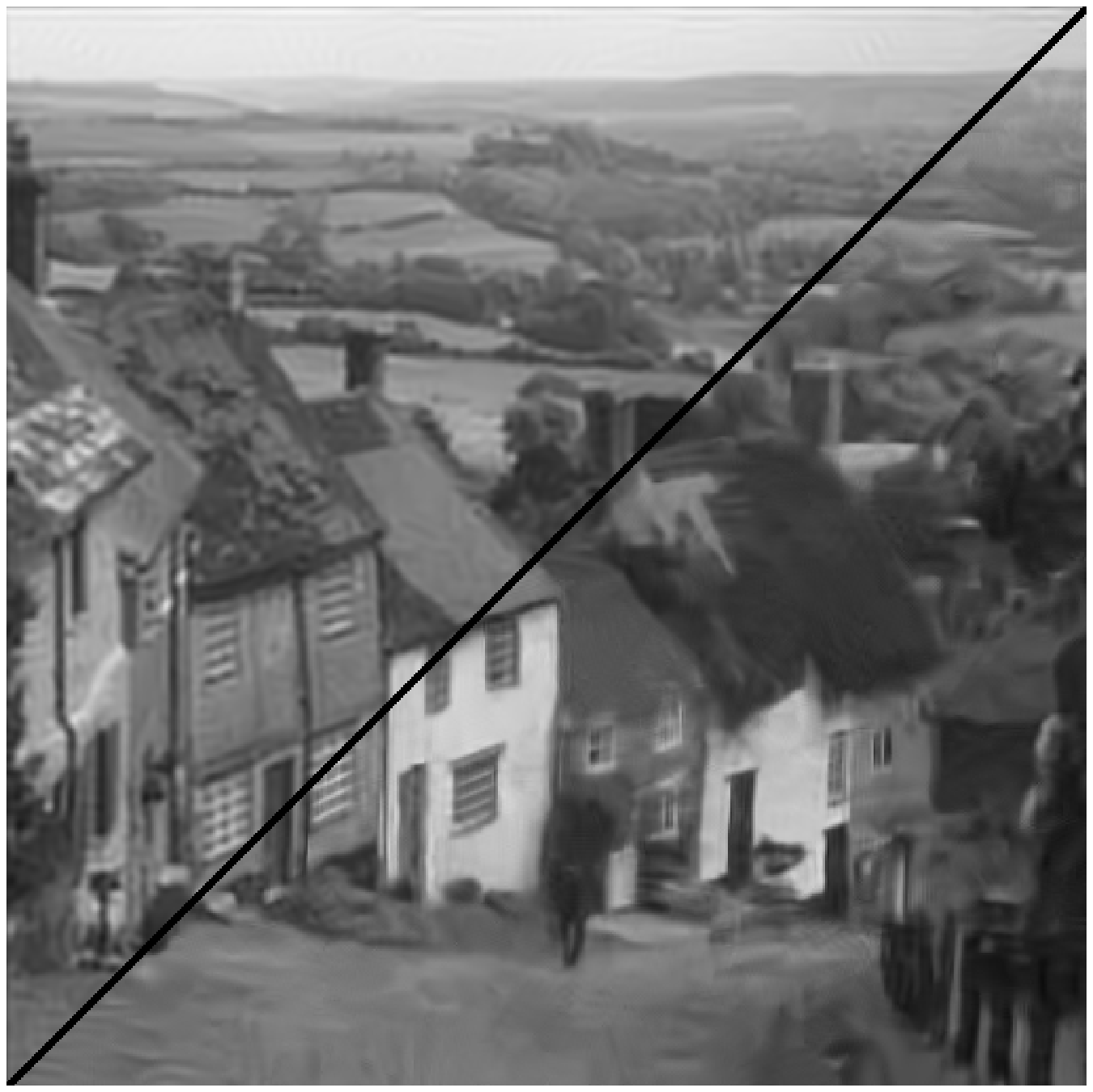} \label{fig:BM3Dy_BM3Dz}}\\
\subfloat[]{\includegraphics[width=0.33\textwidth, trim=2cm 1.5cm 2cm 0.5cm, clip]{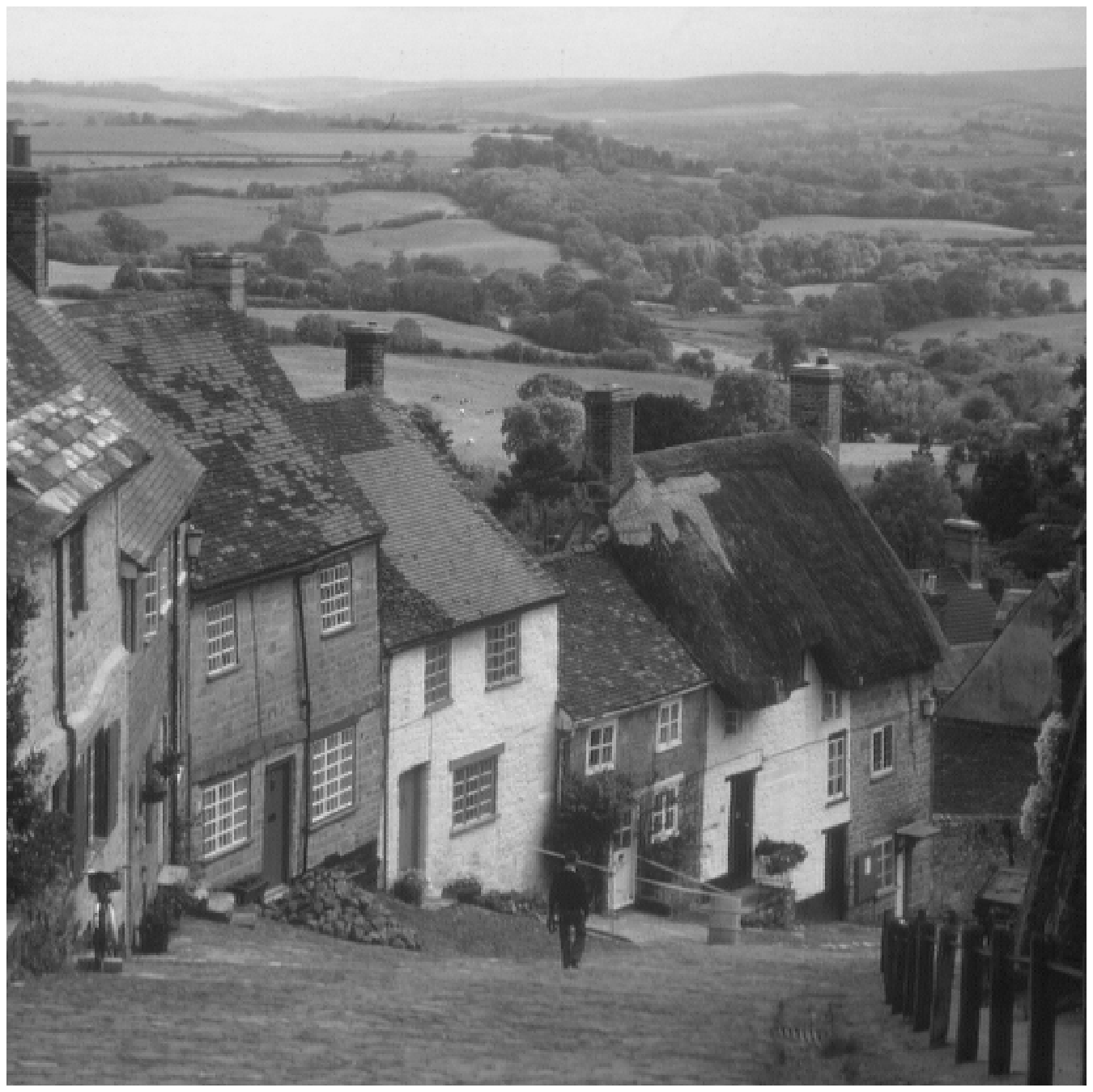} \label{fig:X}}%\hspace{0.2cm}
\subfloat[]{\includegraphics[width=0.33\textwidth, trim=2cm 1.5cm 2cm 0.5cm, clip]{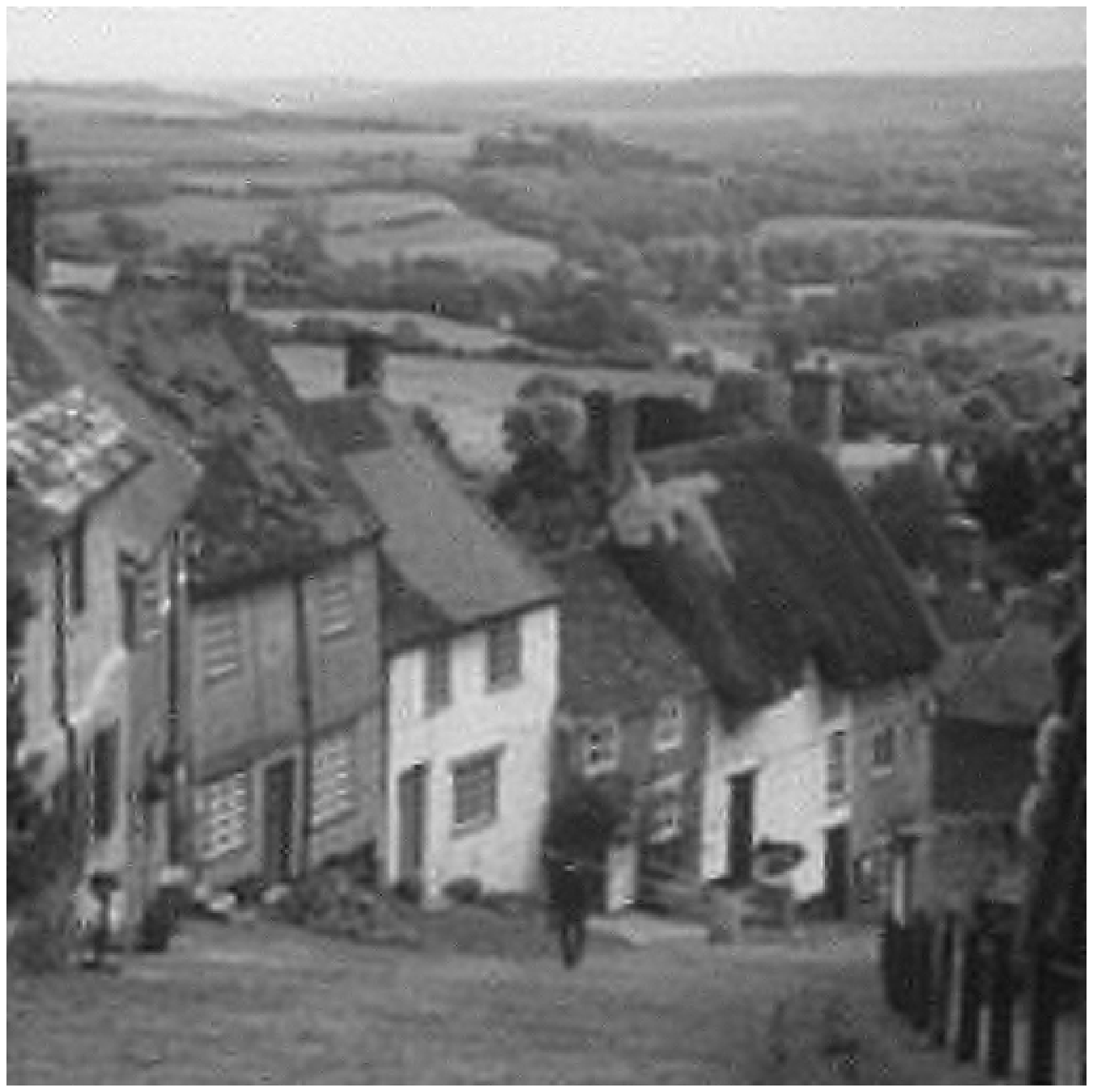} \label{fig:PLMMSEyz}}%\hspace{0.2cm}
\subfloat[]{\includegraphics[width=0.33\textwidth, trim=2cm 1.5cm 2cm 0.5cm, clip]{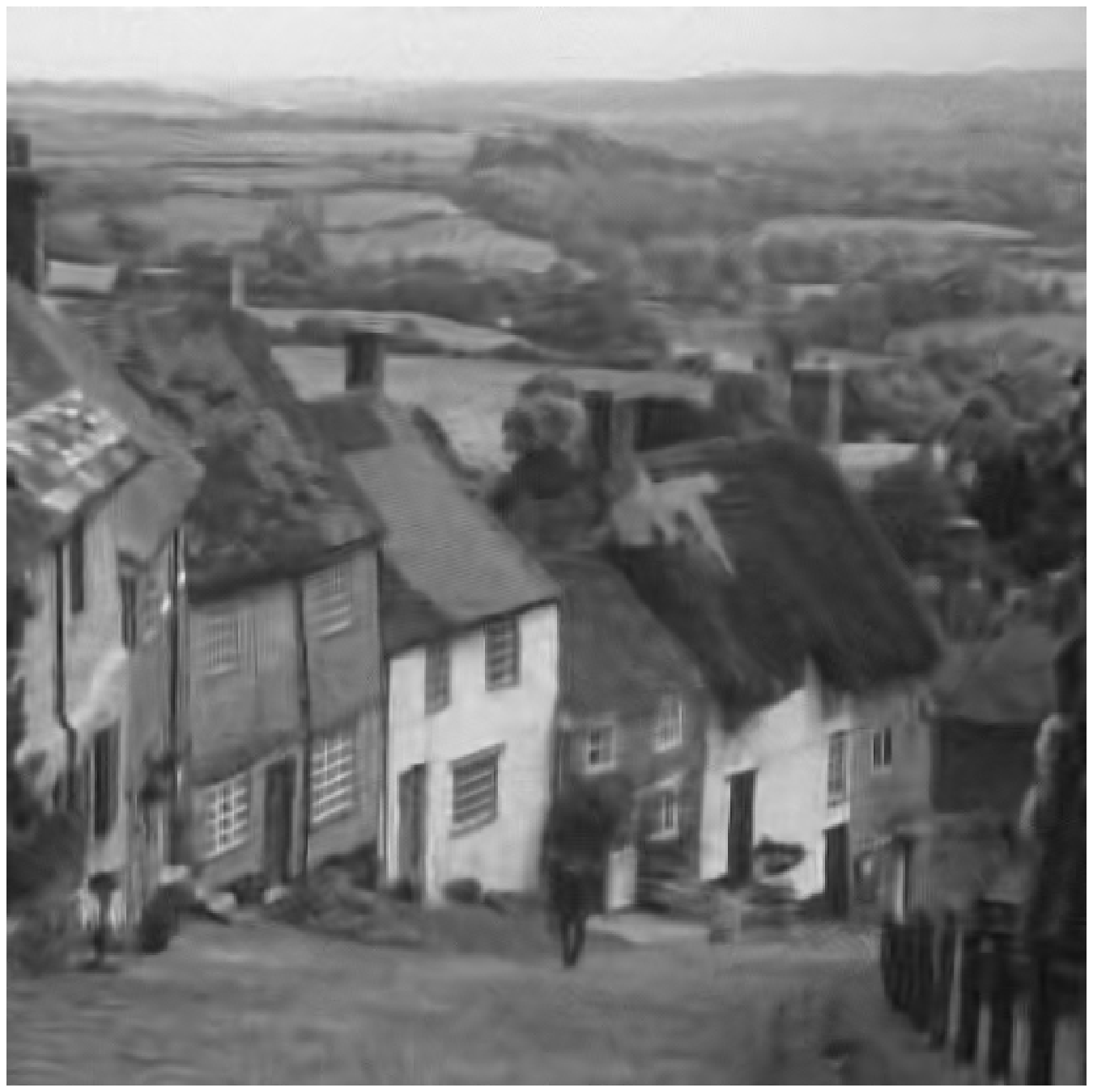} \label{fig:YUANyz}}
\caption{Debluring with a blurred/noisy image pair using PLMMSE estimation and RD \cite{YSQS07}. (a)~Blurred image $Y$ (top left) and noisy image $Z$ (bottom-right). (b)~LMMSE-deblurred image $\hat{X}_Y^{\rm L}$ (top-left) and MMSE-denoised image $\hat{X}_Z^{\rm NL}$ (bottom-right). (c)~BM3D-deblurred image (top left) and BM3D-denoised image (bottom-right). (d)~Original image $X$. (e)~PLMMSE estimate $\hat{X}^{\rm PL}$ from $Y$ and $Z$. (f)~RD recovery from $Y$ and $Z$.}
\label{fig:deblur}
\end{figure*}

%
%\begin{figure*}[t]\centering
%\subfloat[]{\includegraphics[width=0.33\textwidth, trim=2cm 1.5cm 2cm 0.5cm, clip]{peppers_X.eps} \label{fig:X}}%\hspace{0.2cm}
%\subfloat[]{\includegraphics[width=0.33\textwidth, trim=2cm 1.5cm 2cm 0.5cm, clip]{peppers_Y.eps} \label{fig:Y}}%\hspace{0.2cm}
%\subfloat[]{\includegraphics[width=0.33\textwidth, trim=2cm 1.5cm 2cm 0.5cm, clip]{peppers_Z.eps} \label{fig:Z}}\\
%\subfloat[]{\includegraphics[width=0.33\textwidth, trim=2cm 1.5cm 2cm 0.5cm, clip]{peppers_PLMMSE_YZ.eps} \label{fig:X_PLMMSE}}%\hspace{0.2cm}
%\subfloat[]{\includegraphics[width=0.33\textwidth, trim=2cm 1.5cm 2cm 0.5cm, clip]{peppers_LMMSE_Y.eps} \label{fig:X_LMMSE}}%\hspace{0.2cm}
%\subfloat[]{\includegraphics[width=0.33\textwidth, trim=2cm 1.5cm 2cm 0.5cm, clip]{peppers_MMSE_Z.eps} \label{fig:X_MMSE}}
%\caption{Debluring with a blurred/noisy image pair using the PLMMSE method. (a)~Original image $X$. (b)~Blurred image $Y$. (c)~Noisy image $Z$. (d)~PLMMSE estimate $\hat{X}_{\rm PLMMSE}$. (e)~LMMSE deblurred version $\hat{X}_Y^{\rm L}$ of $Y$. (d)~MMSE denoised version $\hat{X}^Z$ of $Z$.}
%\label{fig:deblur}
%\end{figure*}

\begin{table*}
\caption{Performance of deblurring/denoising on several images. Numbers on the left and right of the slash indicate, respectively, PSNR in dB and running time in seconds.}
\label{tbl:comparison}
\centering
\begin{tabular}{l||c|c|c|c|c|c|}
                           & $\hat{X}_Z^{\rm NL}$  & $\hat{X}^{\rm L}_Y$ & BM3D Denoising    & BM3D Deblurring   & PLMMSE           & RD \\ \hline\hline
Boat ($512\times 512$)     & $25.39$ / $0.83$ & $23.45$ / $0.06$    & $27.85$ / $13.52$ & $28.40$ / $10.23$ & $28.05$ / $0.88$ & $29.22$ / $15.31$ \\
Lena ($512\times 512$)     & $26.93$ / $0.73$ & $24.59$ / $0.03$    & $29.47$ / $13.22$ & $30.58$ / $8.90$ & $30.58$ / $0.81$ & $31.37$ / $15.19$ \\
Mandrill ($512\times 512$) & $21.40$ / $0.64$ & $20.59$ / $0.06$    & $22.72$ / $13.58$ & $21.78$ / $9.57$ & $22.58$ / $0.72$ & $23.30$ / $15.58$ \\
Peppers ($512\times 512$)  & $26.74$ / $0.81$ & $24.89$ / $0.08$    & $29.49$ / $13.14$ & $29.74$ / $8.91$ & $29.80$ / $0.88$ & $31.52$ / $15.03$ \\
Mountain ($640\times 480$) & $19.23$ / $0.95$ & $17.69$ / $0.09$    & $20.11$ / $15.24$ & $18.45$ / $11.12$ & $20.03$ / $1.05$ & $20.42$ / $17.47$ \\
Frog ($621\times 498$)     & $23.23$ / $0.94$ & $22.35$ / $0.16$    & $24.00$ / $16.07$ & $24.40$ / $13.37$ & $24.69$ / $1.09$ & $24.69$ / $21.14$\\
Gold-hill ($512\times 512$)& $25.90$ / $0.69$ & $24.26$ / $0.06$    & $27.52$ / $13.41$ & $28.70$ / $9.54$  & $28.82$ / $1.09$ & $29.09$ / $21.14$\\
\hline
Average                    & $24.12$ / $0.81$ & $22.55$ / $0.08$    & $25.88$ / $14.03$ & $26.01$ / $10.23$ & $26.31$ / $0.89$ & $27.09$ / $16.19$
\end{tabular}
\end{table*}

%%%%%%%%%%%%%%%%%%%%%%%%%%%%%%%%%%%%%%%%%%%%%%%%%%%%%%%%%%%%%%%%%%%%%%%%%%%%%%%%%%%%%%%%%%%%%%%%%%%%%%%%%%%%%%%%%%%%%%%%%%%%%%%%%%
%%%%%%%%%%%%%%%%%%%%%%%%%%%%%%%%%%%%%%%%%%%%%%%%%%%%%%%%%%%%%%%%%%%%%%%%%%%%%%%%%%%%%%%%%%%%%%%%%%%%%%%%%%%%%%%%%%%%%%%%%%%%%%%%%%
\section{Application to Maneuvering Target Tracking}
\label{sec:tracking}

Next, we demonstrate PLMMSE estimation in the context of maneuvering target tracking. Applications in which there is a need to track the kinematic features of a target are ubiquitous. Often, multiple types of measurements are available. In non-cooperative scenarios, these may include range, bearing, elevation, range rate (Doppler), and more \cite{LJ01}. In navigation applications, measurements may include the signals of global navigation satellite systems and inertial sensors (accelerometers and rate gyros). Such sensors are used in satellites as well as in modern cellular phones, tablet computers and vehicles. Measurements of this type can be fused to aid, \eg autonomous navigation \cite{SND99} or traffic monitoring \cite{MPR08}.

%
%Modern cellular phones, tablet computers and vehicles are equipped with GPS receivers, measuring the position of the device, as well as with accelerometers and rate gyros, which measure its acceleration. These measurements can be fused to aid, \eg autonomous navigation \cite{SND99} or traffic monitoring \cite{MPR08}.

To model the tracking problem one usually defines a state vector $X(k)$ comprising the target kinematic data, which evolves via the following stochastic linear equation
\begin{align}\label{eq:stateEv}
X(k+1) = \bF_k X(k) + \bB_k W(k).
\end{align}
Here, $\{W(k)\}$ is a zero-mean white noise sequence satisfying $\Cov(W(k))=\sigma_W^2\bI$ for all $k$ and $\{\bF_k\}$ and $\{\bB_k\}$ are known deterministic matrices. For the simplicity of the exposition, we assume that $X(0)=0$ (modification to other initializations is trivial). Suppose that two sets of measurements of the state are observed, so that
\begin{align}\label{eq:measEv}
\begin{pmatrix} Y(k) \\ Z(k) \end{pmatrix} &= \begin{pmatrix}\bH_k \\ \bG_k\end{pmatrix} X(k) + \begin{pmatrix}U(k) \\ V(k)\end{pmatrix},
\end{align}
where $\{U(k)\}$ and $\{V(k)\}$ are mutually independent zero-mean white noise sequences satisfying $\Cov(U(k))=\sigma_U^2\bI$ and $\Cov(V(k))=\sigma_V^2\bI$, and $\{\bH_k\}$ and $\{\bG_k\}$ are given matrices.

At the $n$th time instant, the goal is to obtain an estimate $\hat{X}(n)$ of $X(n)$ based on the measurements $\{Y(k)\}_{k=1}^n$ and $\{Z(k)\}_{k=1}^n$. Ideally, we would like our estimation scheme to possess a recursive structure such that $\hat{X}(n)$ is computed from the previous estimate $\hat{X}(n-1)$ and the current measurements $Y(n)$ and $Z(n)$ without needing to store the entire measurement history.

A simple, yet popular method for modeling maneuvering targets is the dynamic multiple-model method~\cite{Trk_Nav2001YBS} in which $W(k)$ follows a Gaussian mixture distribution. In this case, low intensity noise represents the nominal, non-maneuvering motion regime of the target, and high intensity process noise represents abrupt maneuvers characterized by increased model uncertainty, and caused by, \eg faults in the actuators of an autonomous aerospace system. Unfortunately, the MMSE solution does not admit a recursive implementation~\cite{AMGC02} in this setting.

One alternative is to resort in these cases to LMMSE estimation, whose recursive implementation is given by the Kalman filter. Another option is to employ approximations of the MMSE estimate, which can be computed recursively, such as the interacting-multiple-model (IMM) filter \cite{BB88}. The performance of these methods tends to depend heavily on the assumption that the measurement noises are Gaussian. When their actual distribution is unknown, their performance deteriorates.

Sometimes, nonetheless, the MMSE estimate can be calculated in an online manner. As shown below, this happens, \eg when the state evolves according to the \emph{white acceleration model} \cite{LJ03} and available are acceleration measurements. When supplied with two sets of measurements, only one of which allowing recursive MMSE estimation, it may be advantageous to use PLMMSE estimation rather than approximate MMSE solutions. In this case, under some mild conditions, the PLMMSE estimate can be updated recursively, similarly to the Kalman filter.

Suppose that the distribution of $\{V(k)\}$ and $\{W(k)\}$ is such that, for any $k<n$, $\EE[W(k)|Z(1),\ldots,Z(n)]=\EE[W(k)|Z(k+1)]$ and that the RVs $\{\EE[W(k)|Z(k+1)]\}$ are uncorrelated. As we discuss in the sequel, this implies that the MMSE estimate $\hat{X}_Z^{\rm NL}(n)$ can be computed recursively from $\hat{X}_Z^{\rm NL}(n-1)$ and $Z(n)$. Our goal is to compute the PLMMSE estimate $\hat{X}^{\rm PL}(n)$ of $X(n)$ from $\{Y(k)\}_{k=1}^n$ and $\{Z(k)\}_{k=1}^n$. To obtain a recursive implementation, it is insightful to examine first the batch PLMMSE solution. To this end, we define
\begin{align}
X &= \begin{pmatrix}X(1) \\ \vdots \\ X(n)\end{pmatrix},\quad Y = \begin{pmatrix}Y(1) \\ \vdots \\Y(n)\end{pmatrix},\quad Z = \begin{pmatrix}Z(1) \\\vdots \\ Z(n)\end{pmatrix}, \nonumber\\
U &= \begin{pmatrix}U(1) \\ \vdots \\ U(n)\end{pmatrix},\quad V = \begin{pmatrix} V(1) \\ \vdots \\ V(n)\end{pmatrix},\quad W = \begin{pmatrix}W(0) \\ \vdots \\ W(n-1)\end{pmatrix}.
\end{align}
Therewith, we have from~\eqref{eq:stateEv} that
\begin{align}
X = \bPsi W, \quad Y = \bH X + U, \quad Z = \bG X + V,
\end{align}
where
\begin{align}
\bPsi =
\begin{pmatrix}
\bB_0                                  & 0                                      & \cdots & 0 \\
\bF_1\bB_0                             & \bB_1                                  & \cdots & 0 \\
\bF_2\bF_1\bB_0                        & \bF_2\bB_1                             & \cdots & 0 \\
\vdots                                 & \vdots                                 & \ddots & \vdots\\
\prod_{k=1}^{n-1}\limits\bF_{n-k}\bB_0 & \prod_{k=1}^{n-2}\limits\bF_{n-k}\bB_1 & \cdots & \bB_{n-1}
\end{pmatrix},
\end{align}
\begin{align}
\bH = {\rm diag}\begin{pmatrix} \bH_1  & \bH_2  & \cdots &  \bH_n \end{pmatrix},
\end{align}
and
\begin{align}
\bG = {\rm diag}\begin{pmatrix} \bG_1  & \bG_2  & \cdots &  \bG_n \end{pmatrix}.
\end{align}
Therefore, as we have seen in \eqref{eq:PLEadditiveNoise}, the batch PLMMSE estimate of $X$ from $Y$ and $Z$ is given in this case by
\begin{equation}
\EE[X|Z] + \bA(Y - \bH\EE[X|Z])
\end{equation}
with the matrix $\bA$ of \eqref{eq:PLEadditiveNoiseA_app}.

Noting that $\EE[X|Z]=\bPsi\EE[W|Z]$ and taking into account our assumption that $\EE[W(k-1)|Z]=\EE[W(k-1)|Z(k)]$ and the block lower-triangular structure of $\bPsi$, we see that the $k$th element in the vector $\EE[X|Z]$ is given by
\begin{equation}
\hat{X}_Z^{\rm NL}(k) = \bF_{k-1}\hat{X}_Z^{\rm NL}(k-1) + \bB_{k-1}\EE[W(k-1)|Z(k)].
\end{equation}
We have thus obtained a recursive computation of $\hat{X}_Z^{\rm NL}(k)$ from $\hat{X}_Z^{\rm NL}(k-1)$ and $Z(k)$.

Next, it remains to determine whether the operation of $\bA$ admits a recursive implementation. Before we do so, we recall that the LMMSE estimate of $X$ from $Y$, which is given by
\begin{align}\label{eq:LMMSE_Kalman}
\hat{X}^Y_{\rm L} = \bGa_{XX} \bH^T \left(\bH\bGa_{XX}\bH^T + \sigma_U^2\bI\right)^\dag Y
\end{align}
in this case, can be implemented recursively via the Kalman filter. In our setting, $\bGa_{XX} = \sigma_W^2\bPsi\bPsi^T$ and, due to the assumption that $\Cov(\EE[W|Z]) = \beta\bI$, we have that $\bGa_{\hat{X}_Z^{\rm NL}\hat{X}_Z^{\rm NL}} = \beta\bPsi\bPsi^T = ({\beta}/{\sigma_W^2})\bGa_{XX}$. Therefore, the matrix $\bA$ of the PLMMSE estimate reduces from \eqref{eq:PLEadditiveNoiseA} to
\begin{align}%\label{eq:PLEadditiveNoiseA}
&\left(1-\frac{\beta}{\sigma_W^2}\right)\!\bGa_{XX} \bH^T \left(\!\left(1-\frac{\beta}{\sigma_W^2}\right)\!\bH\bGa_{XX}\bH^T + \sigma_U^2\bI\right)^\dag \nonumber\\
&\hspace{2cm} = \bGa_{XX} \bH^T \left(\bH\bGa_{XX}\bH^T + \frac{\sigma_W^2\sigma_U^2}{\sigma_W^2-\beta}\bI\right)^\dag.
\end{align}
We see that this matrix is the same as that appearing in \eqref{eq:LMMSE_Kalman}, except for the noise variance which is multiplied here by $\sigma_W^2/(\sigma_W^2-\beta)$. This implies that multiplication by $\bA$ corresponds to a Kalman filter with higher observation noise.

We conclude that the complete recursive PLMMSE implementation comprises the following steps:
\begin{enumerate}
\item[a)] {\it Initialization:} $\hat{X}_Z^{\rm NL}(0)=0$, $\tilde{X}(0)=0$, $\bP_{0}=0$.
\item[b)] {\it Recursion:} For $k =1,2,\ldots$ perform the routine summarized in Alg.~\ref{alg:Recursive}.
%Here ${\rm KF}(\hat{X}_{k-1},\bP_{k-1},{Y}_k,\sigma_U^2)$ denotes a one-step Kalman update of the estimate $\hat{X}_{k-1}$ and error-covariance $\bP_{k-1}$ using measurement ${Y}_k$ with noise variance $\sigma_U^2$.
\end{enumerate}

\renewcommand{\algorithmicrequire}{\textbf{Input:}}
\renewcommand{\algorithmicensure}{\textbf{Output:}}
\begin{algorithm}[t]\caption{One Cycle of the Recursive PLMMSE.}
\begin{algorithmic}[1]
\REQUIRE{$Y(k),Z(k),\hat{X}_Z^{\rm NL}(k-1),\tilde{X}(k-1),\bP_{k-1}$}
\STATE{$\hat{W}_Z^{\rm NL}(k-1) = \EE[W(k-1)|Z(k)]$.}
\STATE{$\hat{X}_Z^{\rm NL}(k) = \bF_{k-1}\hat{X}_Z^{\rm NL}(k-1) + \bB_{k-1}\hat{W}_Z^{\rm NL}(k-1)$.}
\STATE{$\hat{Y}_Z^{\rm NL}(k) = \bH_k\hat{X}_Z^{\rm NL}$.}
\STATE{$\tilde{Y}(k) = Y(k) - \hat{Y}_Z^{\rm NL}(k)$.}
\STATE{Update $\tilde{X}(k-1),\bP_{k-1}$ using $\tilde{Y}(k)$ via a Kalman step:
\begin{align*}
\bP_{k}^-&=\bF_{k-1}\bP_{k-1}\bF_{k-1}+\sigma_W^2\bB_k\bB_k^T\\
\bK_{k}&=\bP_{k}^-\bH_k^T\left(\bH\bP_{k}^-\bH_k^T+\frac{\sigma_W^2\sigma_U^2}{\sigma_W^2-\beta}\bI\right)^\dag\\
\bL_k&=(\bI-\bK_k\bH_k)\bA_{k-1}\\
\tilde{X}(k)&=\bL_k\tilde{X}(k-1)+\bK_k\tilde{Y}(k)\\
\bP_k&=\bP_{k}^-\bK_k\left(\bH_k\bP_{k}^-\bH_k^T+\frac{\sigma_W^2\sigma_U^2}{\sigma_W^2-\beta}\bI\right)\bK_k^T
\end{align*}
}
\STATE{$\hat{X}^{\rm PL}(k) = \hat{X}_Z^{\rm NL}(k) + \tilde{X}(k)$.}
\ENSURE{$\hat{X}^{\rm PL}(k),\hat{X}_Z^{\rm NL}(k),\tilde{X}(k),\bP_k$}
\end{algorithmic}
\label{alg:Recursive}
\end{algorithm}

%%%%%%%%%%%%%%%%%%%%%%%%%%%%%%%%%%%%%%%%%%%%%%%%%%%%%%%%%%%%%%%%%%%%%%%%%%%%%%%%%%%%%%%%%%%%%%%%%%%%%%%%%%%%%%%%%%%%%%%%%%%%%%%%%%
\subsection{Example: Tracking a Maneuvering Target From Position and Acceleration Observations}

A common way of describing target kinematics is via the \emph{nearly-constant-velocity model}, also known as the \emph{white acceleration model} \cite{LJ03}. Focusing on one dimension for simplicity, and denoting by $P(k)$ the position of the target at time $kT$, the state vector $X(k) = \begin{pmatrix}P(k) & P(k-1) & P(k-2)\end{pmatrix}^T$ in this model evolves according to \eqref{eq:stateEv} with
\begin{equation}
\bF_k = \begin{pmatrix} 2 & -1 & 0 \\ 1 & 0 & 0 \\ 0 & 1 & 0 \end{pmatrix}.
\end{equation}
If the sampling interval $T$ is small, then, to close extent, $\bB_k=\begin{pmatrix}1 & 0 & 0\end{pmatrix}^T$ (see \eg \cite{LJ03}). The driving noise $W_k$ in this model corresponds to the target's acceleration.

In many applications~\cite{sigalov:oshman:fusion11:actuator}, the target's velocity changes gradually most of the time apart for abrupt transients, which occur every once in a while. This behavior can be described by letting $W_k = S_kB_k$, where $B_k\sim\N(0,1)$ and $\PP(S_k=\sigma_1)=1-\PP(S_k=\sigma_2)=p$. If $\sigma_1$ is much larger than $\sigma_2$ and $p$ is small then the mean time between consecutive large acceleration events is large. In the terminology of the multiple model approach mentioned above, each of the components in the Gaussian mixture corresponds to a different dynamic model.

Suppose that we observe noisy measurements $Y(k)$ and $Z(k)$ of the position $P(k)$ and acceleration $W(k-1)$, respectively. These measurements relate to the state vector via \eqref{eq:measEv}, with $\bH_k = \begin{pmatrix} 1 & 0 & 0 \end{pmatrix}$ and $\bG_k = \begin{pmatrix} 1 & -2 & 1 \end{pmatrix}$. Indeed, it is easily verified that in this setting $Z(k) = W(k-1)+V(k)$. Consequently, for $n>k$, $\EE[W(k)|Z(1),\ldots,Z(n)]=\EE[W(k)|Z(k+1)] = f(Z(k+1))$ with the function $f(\cdot)$ of \eqref{eq:f}.

Equipped with the matrices $\bF_k$, $\bB_k$, $\bH_k$, and $\bG_k$, we estimate the state vector of~\eqref{eq:stateEv} using the recursive PLMMSE method described above. Note that although $X(k)$ comprises only the target's positions at times $k$, $k-1$, and $k-2$, estimating the velocity and acceleration is straightforward as $\hat{P}(k)-\hat{P}(k-1)$ and $\hat{P}(k)-2\hat{P}(k-1)+\hat{P}(k-2)$, respectively.

We compare the performance of the recursive PLMMSE method with the one of a standard KF which provides, recursively, the LMMSE optimal estimate of the state using the measurement sets $\{Z(k)\}$ and $\{Y(k)\}$, as well as with that of the IMM filter~\cite{BB88} which is known to be extremely effective in multiple model estimation problems~\cite{Trk_Nav2001YBS}. The main idea underlying the IMM algorithm is to maintain a bank of primitive Kalman filters, each matched to a different model in the given model set. Each filter produces a local estimate with an associated error covariance using its initial estimate and covariance and the current measurement. The key element of the IMM scheme is the interaction block that generates, using all local estimates, individual initial conditions for each of the primitive filters in the bank. In our case, the two models maintained by the IMM method correspond to the two possible realizations of $S_k$ such that one model corresponds to the low process noise variance representing the nominal target regime and the second model for the maneuvering one.

We simulated a random sequence $\{X(k)\}$ for $k=1,\ldots,1000$ according to \eqref{eq:stateEv} initialized at $X(0)=0$ and driven by a process noise $\{W(k)\}$ having, at each time, a two-modal Gaussian mixture distribution, with $p=0.05$, $\sigma_{1}=10$, and $\sigma_{2}=1$. The state position is observed via a Gaussian measurement equation with $\sigma_U=5$ and the covariance of the Gaussian measurement noise of the acceleration, $\sigma_V$, is swept from $1$ to $15$. Averaged over $500$ independent Monte Carlo runs, the average squared position, velocity, and acceleration errors are presented, respectively, on the top, middle, and bottom of the left chart of Fig.~\ref{fig:track}. Outperformed by the IMM, the PLMMSE method scores better than the Kalman filter since it optimally utilizes the acceleration measurements. It is noticeable that at high values of $\sigma_V^2$, Kalman's and PLMMSE's errors coincide indicating that acceleration measurements do not carry valuable information in addition to that carried by $\{Y(k)\}$.

In many practical scenarios the distribution of the position measurement noise is far from Gaussian (see \eg \cite{IS10} and references therein). On the right chart of Fig.~\ref{fig:track} we present the position, velocity and acceleration errors obtained for a Gaussian mixture distribution of the position noise having the same first- and second-order statistics as before. Such a distribution may model occasional outlier measurements or sensor faults~\cite{mazor1998interacting}. None of the filters is supplied with this information and only the first- and second-order moments are provided to the algorithms. Utilizing the position measurements in a linear manner, both Kalman and PLMMSE keep the performance unchanged relative to the Gaussian case. In a contradistinction, the IMM algorithm results in an inferior performance. This phenomenon is tightly related to the statement of Theorem~\ref{thm:PLMMSEminimax} claiming that the PLMMSE method is ensured to attain a smaller worst-case MSE in comparison to any other estimation technique provided the appropriate moments are kept constant.

%width=0.5\textwidth

\begin{figure*}[t]\centering
\subfloat[]{\includegraphics[scale=0.6]{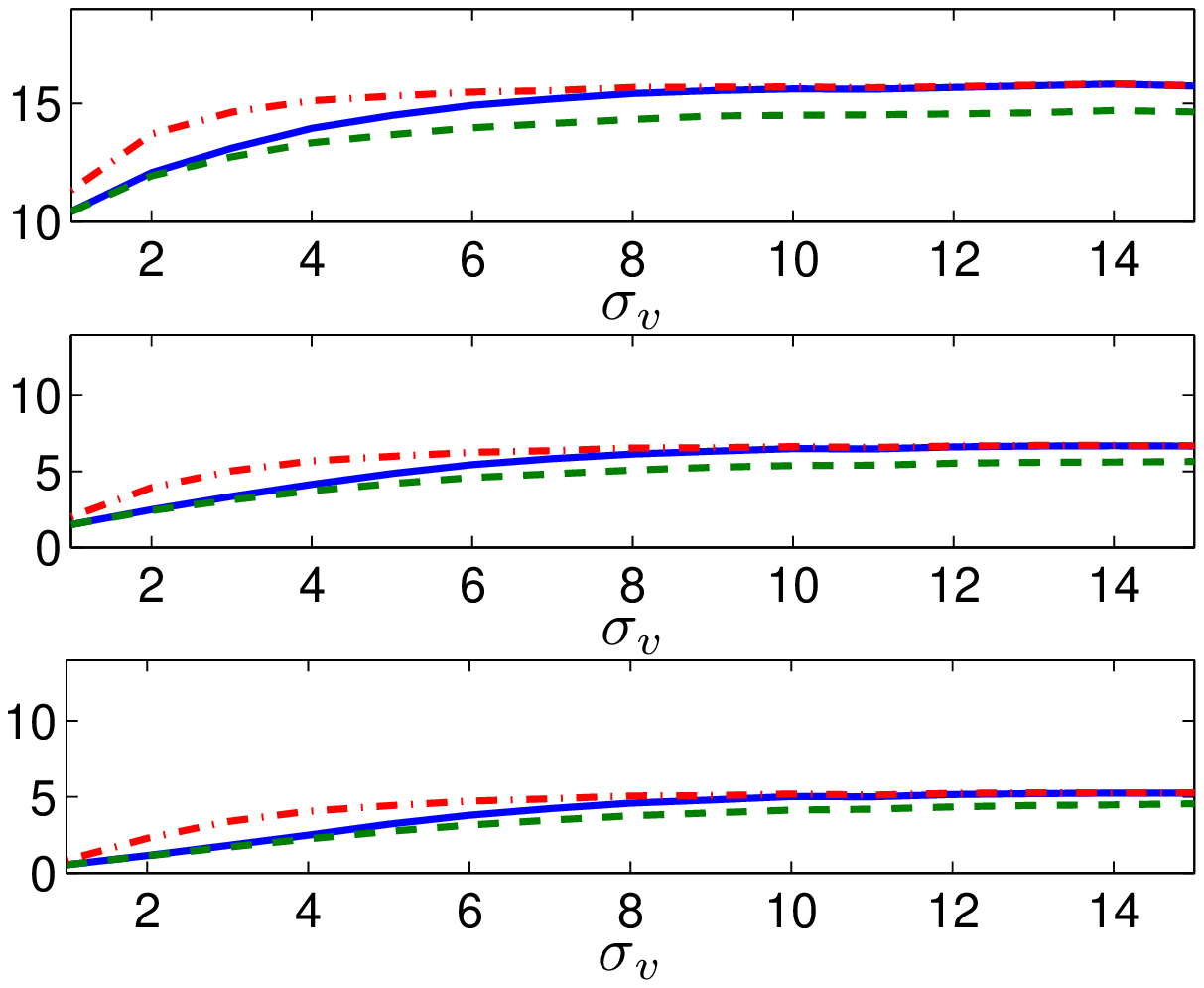} \label{fig:GAUSS}}
\subfloat[]{\includegraphics[scale=0.6]{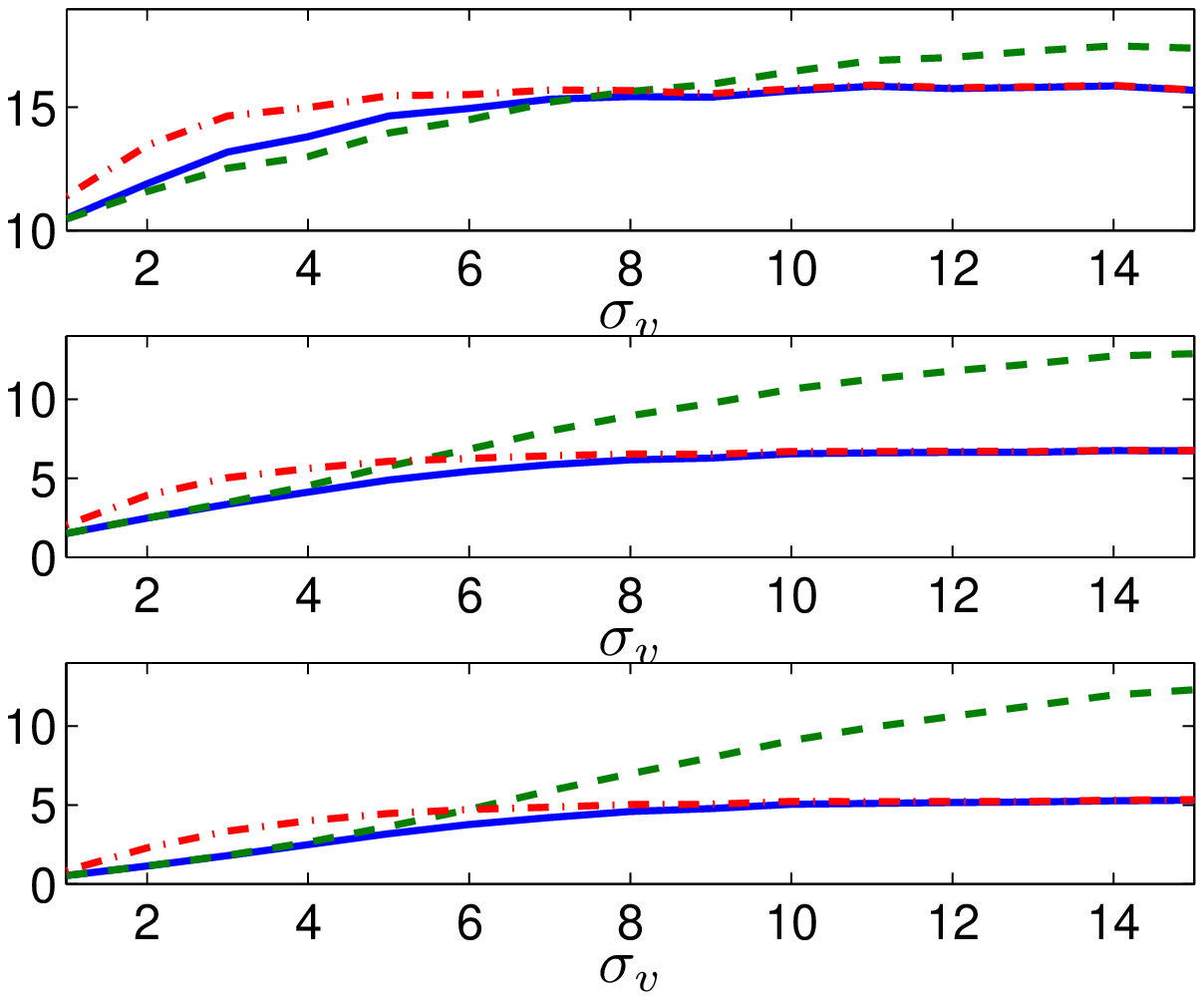} \label{fig:GM}}%\hspace{0.2cm}
\caption{Mean squared estimation errors of the position (top), velocity (middle) and acceleration (bottom) vs. $\sigma_v$ using the PLMMSE approach (solid), IMM (dashed) and standard KF (dash-dotted). (a) $\{U(k)\}$ have a Gaussian distribution. (b)  $\{U(k)\}$ have a Gaussian mixture distribution.}
\label{fig:track}
\end{figure*}

%%%%%%%%%%%%%%%%%%%%%%%%%%%%%%%%%%%%%%%%%%%%%%%%%%%%%%%%%%%%%%%%%%%%%%%%%%%%%%%%%%%%%%%%%%%%%%%%%%%%%%%%%%%%%%%%%%%%%%%%%%%%%%%%%%
%%%%%%%%%%%%%%%%%%%%%%%%%%%%%%%%%%%%%%%%%%%%%%%%%%%%%%%%%%%%%%%%%%%%%%%%%%%%%%%%%%%%%%%%%%%%%%%%%%%%%%%%%%%%%%%%%%%%%%%%%%%%%%%%%%
\section{Conclusions}
\label{sec:conclusions}
In this paper we derived the PLMMSE estimator, which is the method whose MSE is minimal among all functions that are linear in $Y$. We showed that the PLMMSE solution depends only on the joint second-order statistics of $X$ and $Y$, which renders it applicable in a wide variety of situations. Furthermore, we showed that this estimator attains the lowest worst-case MSE over the set of distributions whose joint second-order moments of $X$ and $Y$ are fixed. We demonstrated our approach in the context of recovering a vector, which is sparse in a unitary dictionary, from a pair of noisy measurements. In this setting, the PLMMSE solution achieves an MSE very close to that attained by iterative approximation strategies, such as the FBMP method of \cite{SPZ08}, and is much cheaper computationally. We applied our method to the problems of image enhancement from blurred/noisy image pairs and maneuvering target tracking from position and acceleration measurements. In both applications, we showed that PLMMSE estimation performs close to state-of-the-art algorithms. In the image enhancement setting, we showed that it can run much faster than competing approaches. In the context of target tracking, we demonstrated the insensitivity of the solution to the distribution of the noise in $Y$. This property provides robustness against sensor faults and outlier measurements, problems which are very common in target tracking situations.

%%%%%%%%%%%%%%%%%%%%%%%%%%%%%%%%%%%%%%%%%%%%%%%%%%%%%%%%%%%%%%%%%%%%%%%%%%%%%%%%%%%%%%%%%%%%%%%%%%%%%%%%%%%%%%%%%%%
%%%%%%%%%%%%%%%%%%%%%%%%%%%%%%%%%%%%%%%%%%%%%%%%%%%%%%%%%%%%%%%%%%%%%%%%%%%%%%%%%%%%%%%%%%%%%%%%%%%%%%%%%%%%%%%%%%%
%%%%%%%%%%%%%%%%%%%%%%%%%%%%%%%%%%%%%%%%%%%%%%%%%%%%%%%%%%%%%%%%%%%%%%%%%%%%%%%%%%%%%%%%%%%%%%%%%%%%%%%%%%%%%%%%%%%
\appendices

%%%%%%%%%%%%%%%%%%%%%%%%%%%%%%%%%%%%%%%%%%%%%%%%%%%%%%%%%%%%%%%%%%%%%%%%%%%%%%%%%%%%%%%%%%%%%%%%%%%%%%%%%%%%%%%%%%%%%%%%%%%%%%%%%%
%%%%%%%%%%%%%%%%%%%%%%%%%%%%%%%%%%%%%%%%%%%%%%%%%%%%%%%%%%%%%%%%%%%%%%%%%%%%%%%%%%%%%%%%%%%%%%%%%%%%%%%%%%%%%%%%%%%%%%%%%%%%%%%%%%
\section{Proof of Theorem~\ref{thm:PL}}
\label{sec:ProofThmPL}
Using the smoothing property, the MSE of any estimator of the form \eqref{eq:PLmodel1} is given by
\begin{equation}
\EE\left[\EE\left[\left\|X- \bA(Z) Y - b(Z)\right\|^2|Z\right]\right].
\end{equation}
Thus, for every specific value $z$ that $Z$ can take, the optimal choice of $\bA(z)$ and $b(z)$ is that minimizing the inner expectation. The solution to this minimization problem corresponds to the LMMSE estimate of $X$ based on $Y$, under the the joint distribution of $(X,Y)$ given $Z$, concluding the proof.

\section{Proof of Theorem~\ref{thm:SPL}}
\label{sec:ProofThmSPL}
We start by noting that the set $\mathcal{B}$ of RVs constituting candidate estimates is a closed linear subspace within the space of finite-second-order-moment RVs taking values in $\RN^M$. Therefore, the MMSE estimate $\hat{X}$ within this subspace, which is the projection of the RV $X$ onto $\mathcal{B}$, is the unique\footnote{In an almost-sure sense.} RV whose estimation error $\hat{X}-X$ is orthogonal to every RV of the form $\bA Y+b(Z)$. To demonstrate that $\hat{X}$ of \eqref{eq:Xhat} indeed satisfies this property, note that the inner product between $\hat{X}-X$ and $\bA Y+b(Z)$ is given by
\begin{align}\label{eq:proof1}
\EE\left[(\hat{X}-X)^T(\bA Y + b(Z))\right] &= \Tr\left\{\EE\left[(\hat{X}-X)Y^T\right]\bA^T\right\}\nonumber\\
 &\hspace{0.4cm} +\Tr\left\{\EE\left[(\hat{X}-X)b(Z)^T\right]\right\}.
\end{align}
Substituting \eqref{eq:Xhat}, the expectation within the second term becomes
\begin{align}
&\EE\left[\left(\bGa_{X\tilde{Y}}\bGa_{\tilde{Y}\tilde{Y}}^\dag \tilde{Y} + \EE[X|Z]-X\right)b(Z)^T\right]  \nonumber\\
&\hspace{0.25cm}=\bGa_{X\tilde{Y}}\bGa_{\tilde{Y}\tilde{Y}}^\dag\EE\left[\tilde{Y}b(Z)^T\right] + \EE\left[\left(\EE[X|Z]-X\right)b(Z)^T\right].
\end{align}
Recall that $\tilde{Y}=Y-\EE[Y|Z]$ is the estimation error incurred in estimating $Y$ from $Z$. Consequently, $\tilde{Y}$ and $X-\EE[X|Z]$ are uncorrelated with every function of $Z$ and, in particular, with $b(Z)$, so that this expression vanishes.
Similarly, substituting \eqref{eq:Xhat} and expressing $Y = \tilde{Y} + \EE[Y|Z]$, the expectation within the first summand in \eqref{eq:proof1} becomes
\begin{align}
&\EE\left[\left(\bGa_{X\tilde{Y}}\bGa_{\tilde{Y}\tilde{Y}}^\dag \tilde{Y} + \EE[X|Z]-X\right)Y^T\right]  \nonumber\\
%&\hspace{1cm}\bGa_{X\tilde{Y}}\bGa_{\tilde{Y}\tilde{Y}}^\dag\EE\left[(Y - \EE[Y|Z])Y^T\right] + \nonumber\\
%&\hspace{4.25cm}\EE\left[\left(\EE[X|Z]-X\right)Y^T\right] = \nonumber\\
&\hspace{1cm}=\bGa_{X\tilde{Y}}\bGa_{\tilde{Y}\tilde{Y}}^\dag \EE\left[\tilde{Y}(\tilde{Y}+\EE[Y|Z])^T\right]\nonumber\\
&\hspace{1.4cm}-\EE\left[(X-\EE[X|Z])(\tilde{Y}+\EE[Y|Z])^T\right].
\end{align}
Being a function of $Z$, the RV $\EE[Y|Z]$ is uncorrelated with $\tilde{Y}$ and $X-\EE[X|Z]$ so that this expression can be reduced to
\begin{align}
&\bGa_{X\tilde{Y}}\bGa_{\tilde{Y}\tilde{Y}}^\dag \EE\left[\tilde{Y}\tilde{Y}^T\right]-\EE\left[(X-\EE[X|Z])W^T\right]\nonumber\\
&\hspace{0.3cm}=\bGa_{X\tilde{Y}}\bGa_{\tilde{Y}\tilde{Y}}^\dag \bGa_{\tilde{Y}\tilde{Y}}-\EE\left[(X-\mu_X+\mu_X-\EE[X|Z])\tilde{Y}^T\right]\nonumber\\
&\hspace{0.3cm}=\bGa_{X\tilde{Y}}-\bGa_{X\tilde{Y}}+\EE\left[(\EE[X|Z]-\mu_X)\tilde{Y}^T\right]\nonumber\\
&\hspace{0.3cm}=\bGa_{X\tilde{Y}}-\bGa_{X\tilde{Y}} \nonumber\\
&\hspace{0.3cm}= 0,
\end{align}
where we used the facts that $\EE[\tilde{Y}]=0$, that $\bGa_{X\tilde{Y}}\bGa_{\tilde{Y}\tilde{Y}}^\dag \bGa_{\tilde{Y}\tilde{Y}}=\bGa_{X\tilde{Y}}$ \cite[Lemma 2]{TH02}, and that $\tilde{Y}$ is uncorrelated with $\EE[X|Z]$ (due to the same argument as above). This completes the proof.

%%%%%%%%%%%%%%%%%%%%%%%%%%%%%%%%%%%%%%%%%%%%%%%%%%%%%%%%%%%%%%%%%%%%%%%%%%%%%%%%%%%%%%%%%%%%%%%%%%%%%%%%%%%%%%%%%%%%%%%%%%%%%%%%%%
%%%%%%%%%%%%%%%%%%%%%%%%%%%%%%%%%%%%%%%%%%%%%%%%%%%%%%%%%%%%%%%%%%%%%%%%%%%%%%%%%%%%%%%%%%%%%%%%%%%%%%%%%%%%%%%%%%%%%%%%%%%%%%%%%%
\section{Derivation of Equation \eqref{eq:XhatExplicit}}
\label{sec:DerivationOfXhatExplicit}
By definition,
\begin{align}\label{eq:Ga_XW}
\bGa_{X\tilde{Y}} &= \EE[(X-\mu_X)(Y-\EE[Y|Z])^T] \nonumber\\
&= \EE[(X-\mu_X)(Y-\mu_Y+\mu_Y-\EE[Y|Z])^T] \nonumber\\
&= \EE[(X\! -\!\mu_X)(Y\! -\!\mu_Y)^T] -\EE[(X\! -\!\mu_X)(\EE[Y|Z]\! -\!\mu_Y)^T] \nonumber\\
&= \bGa_{XY} - \EE[\EE[(X-\mu_X)(\EE[Y|Z]-\mu_Y)^T|Z]] \nonumber\\
&= \bGa_{XY} - \EE[(\EE[X|Z]-\mu_X)(\EE[Y|Z]-\mu_Y)^T] \nonumber\\
&= \bGa_{XY} - \bGa_{\hat{X}_Z^{\rm NL}\hat{Y}_Z^{\rm NL}},
\end{align}
where $\hat{X}_Z^{\rm NL}=\EE[X|Z]$ and $\hat{Y}_Z^{\rm NL}=\EE[Y|Z]$. Here, the fourth equality is a result of the smoothing property and the last equality follows from the facts that $\EE[\EE[X|Z]]=\mu_X$ and $\EE[\EE[Y|Z]]=\mu_Y$. In a similar manner, it is easy to show that
\begin{equation}\label{eq:Ga_YW}
\bGa_{Y\tilde{Y}} = \bGa_{YY} - \bGa_{\hat{Y}_Z^{\rm NL}\hat{Y}_Z^{\rm NL}}.
\end{equation}
Using \eqref{eq:Ga_YW} and the fact that $\tilde{Y}$ is uncorrelated with $\EE[Y|Z]-\mu_Y$, we obtain
%the fact that $Y$ is uncorrelated with $Y-\EE[Y|Z]$,
\begin{align}\label{eq:Ga_WW}
\bGa_{\tilde{Y}\tilde{Y}} &= \EE[\tilde{Y}\tilde{Y}^T] \nonumber\\
&= \EE[(Y-\EE[Y|Z])\tilde{Y}^T] \nonumber\\
&= \EE[(Y-\mu_Y)\tilde{Y}^T] - \EE[(\EE[Y|Z]-\mu_Y)\tilde{Y}^T]\nonumber\\
&= \bGa_{Y\tilde{Y}} \nonumber\\
&= \bGa_{YY} - \bGa_{\hat{Y}_Z^{\rm NL}\hat{Y}_Z^{\rm NL}}.
%%&= \EE[(Y-\EE[Y]+\EE[Y]-\EE[Y|Z])(Y-\EE[Y]+\EE[Y]-\EE[Y|Z])^T] - \nonumber\\
%&= \EE[(Y-\EE[Y])(Y-\EE[Y])^T] \nonumber\\
%&\hspace{1.5cm} - \EE[(Y-\EE[Y])(\EE[Y|Z]-\EE[Y])^T] \nonumber\\
%&\hspace{1.5cm} -\EE[(\EE[Y|Z]-\EE[Y])(Y-\EE[Y])^T] \nonumber\\
%&\hspace{1.5cm} +\EE[(\EE[Y|Z]-\EE[Y])(\EE[Y|Z]-\EE[Y])^T] \nonumber\\
%&= \bGa_{YY} - \EE[\EE[(Y-\EE[Y|Z])(\EE[Y|Z]-\EE[Y])^T|Z]] \nonumber\\
%&= \bGa_{XY} - \EE[(\EE[X|Z]-\EE[X])(\EE[Y|Z]-\EE[Y])^T] \nonumber\\
%&= \bGa_{XY} - \bGa_{\hat{X}^Z\hat{X}^Z}.
\end{align}
Substituting \eqref{eq:Ga_XW} and \eqref{eq:Ga_WW} into \eqref{eq:Xhat} leads to \eqref{eq:XhatExplicit}.

%%%%%%%%%%%%%%%%%%%%%%%%%%%%%%%%%%%%%%%%%%%%%%%%%%%%%%%%%%%%%%%%%%%%%%%%%%%%%%%%%%%%%%%%%%%%%%%%%%%%%%%%%%%%%%%%%%%%%%%%%%%%%%%%%%
%%%%%%%%%%%%%%%%%%%%%%%%%%%%%%%%%%%%%%%%%%%%%%%%%%%%%%%%%%%%%%%%%%%%%%%%%%%%%%%%%%%%%%%%%%%%%%%%%%%%%%%%%%%%%%%%%%%%%%%%%%%%%%%%%%
\section{Proof of Theorem \ref{thm:PLMMSEminimax}}
\label{sec:ProofOptPLMMSE}
Let $\varepsilon(F_{XYZ},\hat{X})=\EE_{F_{XYZ}}[\|\hat{X}-X\|^2]$ denote the MSE incurred by an estimator $\hat{X}$ of $X$ based on $Y$ and $Z$, when the joint distribution of $X$, $Y$ and $Z$ is $F_{XYZ}(x,y,z)$. It is easily verified that
\begin{align}\label{eq:MSE_PLMMSE}
&\varepsilon(F_{XYZ},\hat{X}_{\rm PLMMSE}) = \Tr\{\bGa_{XX}\}\nonumber\\
&-\Tr\left\{(\bGa_{XY}\!-\!\bGa_{\hat{X}_Z^{\rm NL}\!\hat{Y}_Z^{\rm NL}})(\bGa_{YY}\!-\!\bGa_{\hat{Y}_Z^{\rm NL}\!\hat{Y}_Z^{\rm NL}})^\dag(\bGa_{XY}\!-\!\bGa_{\hat{X}_Z^{\rm NL}\!\hat{Y}_Z^{\rm NL}})^T\right\}
\end{align}
for all $F_{XYZ}\in\A$. Therefore, in particular, \eqref{eq:MSE_PLMMSE} is also the worst-case MSE of $\hat{X}^{\rm PL}$ over $\A$. We next make use of the following lemma.

\begin{lemma}\label{le:ExistDist}
There exists a distribution $F^*_{XYZ}$ in the set $\A$ of distributions satisfying \eqref{eq:setA}, under which the PLMMSE estimate of $X$ based on $Y$ and $Z$ coincides with the MMSE estimate $\EE[X|Y,Z]$.
\end{lemma}
\begin{proof}
See Appendix~\ref{sec:ProofLemmaExistDist}.
\end{proof}

Now, any estimator $\hat{X}$ that is a function of $Y$ and $Z$ satisfies
\begin{align}
\sup_{F_{XYZ}\in\A}\varepsilon(F_{XYZ},\hat{X}) &\geq \varepsilon(F^*_{XYZ},\hat{X}) \nonumber\\
&\geq \min_{\hat{X}}\varepsilon(F^*_{XYZ},\hat{X}) \nonumber\\
%&= \EE_{f_{XY}}[\|\EE[X|Y]-X\|^2]\nonumber\\
&=\varepsilon(F^*_{XYZ},\EE[X|Y,Z]) \nonumber\\
&=\varepsilon(F^*_{XYZ},\hat{X}^{\rm PL}) \nonumber\\
&=\max_{F_{XYZ}\in\A}\varepsilon(F_{XYZ},\hat{X}^{\rm PL}),
\end{align}
where the first line follows from the fact that $F^*_{XYZ}\in\A$, the third line is a result of the fact that the MMSE and PLMMSE estimators coincide under $F^*_{XYZ}$, and the last line is due to the fact that $\varepsilon(F_{XYZ},\hat{X}^{\rm PL})$ is constant as a function of $F_{XYZ}$ over $\A$. We have thus established that the worst-case MSE of any estimator over $\A$ is greater or equal to the worst-case MSE of the PLMMSE solution over $\A$, proving that $\hat{X}^{\rm PL}$ is minimax optimal.

%%%%%%%%%%%%%%%%%%%%%%%%%%%%%%%%%%%%%%%%%%%%%%%%%%%%%%%%%%%%%%%%%%%%%%%%%%%%%%%%%%%%%%%%%%%%%%%%%%%%%%%%%%%%%%%%%%%%%%%%%%%%%%%%%%
%%%%%%%%%%%%%%%%%%%%%%%%%%%%%%%%%%%%%%%%%%%%%%%%%%%%%%%%%%%%%%%%%%%%%%%%%%%%%%%%%%%%%%%%%%%%%%%%%%%%%%%%%%%%%%%%%%%%%%%%%%%%%%%%%%
\section{Proof of Lemma~\ref{le:ExistDist}}
\label{sec:ProofLemmaExistDist}
We prove the statement by construction.
Let $Y$ and $Z$ be two RVs distributed according to $F_{YZ}$ and denote $h(Z)=\EE[Y|Z]$ and $\tilde{Y}=Y-h(Z)$. Let $U$ be a zero-mean RV, statistically independent of the pair $(\tilde{Y},Z)$, whose covariance matrix is
\begin{align}
\bGa_{UU} &= \bGa_{XX} - \Cov(g(Z)) - \bGa_{X\tilde{Y}}\bGa_{\tilde{Y}\tilde{Y}}^{\dag}\bGa_{\tilde{Y}X}. \label{eq:VproofMinimax}
\end{align}
It can be easily verified that this is the covariance matrix of the estimation error of the LMMSE estimate of $\tilde{X}=X-\EE[X|Z]$ based on $\tilde{Y}=Y-\EE[Y|Z]$. Therefore, this is a valid covariance matrix.
%where
%\begin{equation}\label{eq:GproofMinimax}
%\bGa_{WW} = \bGa_{XY}-\Cov(g(Z),h(Z)).
%\end{equation}
Consider the RV\footnote{Recall that $\bGa_{X\tilde{Y}}$ and $\bGa_{\tilde{Y}\tilde{Y}}$ are functions of $\Cov(X,Y)$ and $F_{YZ}$, which are given.}
\begin{align}
X &= \bGa_{X\tilde{Y}}\bGa_{\tilde{Y}\tilde{Y}}^{\dag} \tilde{Y} + g(Z) + U. \label{eq:XproofMinimax}
%Y &= h(Z) + V.\label{eq:YproofMinimax}
\end{align}

We will show that the so constructed $X$, $Y$ and $Z$ satisfy the constraints \eqref{eq:setA}. Indeed, using the fact that $U$ has zero mean and is statistically independent of $Z$, we find that the conditional expectation of $X$ of \eqref{eq:XproofMinimax} given $Z$ is
\begin{align}\label{eq:EXZ}
\EE[X|Z] &= g(Z).
\end{align}
%and that the conditional expectation of $Y$ of \eqref{eq:YproofMinimax} given $Z$ is
%\begin{equation}\label{eq:EYZ}
%\EE[Y|Z] = h(Z).
%\end{equation}
%Furthermore, the covariance matrix of $Y$ of \eqref{eq:YproofMinimax} is given by
%\begin{align}\label{eq:CovYY}
%\Cov(Y) &= \Cov(h(Z)) + \bGa_{VV} \nonumber\\
%&= \Cov(h(Z)) + \bGa_{YY} - \Cov(h(Z)) \nonumber\\
%&= \bGa_{YY},
%\end{align}
%where we substituted \eqref{eq:UproofMinimax}.
Furthermore, since $\tilde{Y}$, $g(Z)$ and $U$ are pairwise uncorrelated, the covariance of $X$ of \eqref{eq:XproofMinimax} can be computed as
\begin{align}\label{eq:CovXX}
\Cov(X) &= \bGa_{X\tilde{Y}}\bGa_{\tilde{Y}\tilde{Y}}^{\dag}\bGa_{\tilde{Y}\tilde{Y}}\bGa_{\tilde{Y}\tilde{Y}}^{\dag}\bGa_{\tilde{Y}X} + \Cov(g(Z))+\bGa_{UU}\nonumber\\
&= \bGa_{X\tilde{Y}}\bGa_{\tilde{Y}\tilde{Y}}^{\dag}\bGa_{\tilde{Y}X} + \Cov(g(Z)) \nonumber\\
&\hspace{0.4cm} + \bGa_{XX}- \Cov(g(Z))- \bGa_{X\tilde{Y}}\bGa_{\tilde{Y}\tilde{Y}}^{\dag}\bGa_{\tilde{Y}X}\nonumber\\
&= \bGa_{XX},
\end{align}
where we substituted \eqref{eq:VproofMinimax}. Finally, the cross covariance of $X$ of \eqref{eq:XproofMinimax} and $Y$ is given by
\begin{align}\label{eq:CovXY}
\Cov(X,Y) &= \bGa_{X\tilde{Y}}\bGa_{\tilde{Y}\tilde{Y}}^{\dag}\bGa_{\tilde{Y}Y} + \Cov(g(Z),h(Z)) \nonumber\\
&= \bGa_{X\tilde{Y}}\bGa_{\tilde{Y}\tilde{Y}}^{\dag}\bGa_{\tilde{Y}\tilde{Y}} + \Cov(g(Z),h(Z)) \nonumber\\
&= \bGa_{X\tilde{Y}} + \Cov(g(Z),h(Z)) \nonumber\\
&= \bGa_{XY} - \Cov(g(Z),h(Z)) + \Cov(g(Z),h(Z))\nonumber\\
&= \bGa_{XY},
\end{align}
where the second equality follows from the third line of \eqref{eq:Ga_WW}, the third equality follows from \cite[Lemma~2]{TH02}, and the fourth equality follows from \eqref{eq:Ga_XW}. Equations \eqref{eq:EXZ}, \eqref{eq:CovXX} and \eqref{eq:CovXY} demonstrate that the distribution $F_{XYZ}^*$ associated with $X$, $Y$ and $Z$, belongs to the family of distributions $\A$ satisfying \eqref{eq:setA}.

Next, we show that the PLMMSE and MMSE estimators coincide under $F_{XYZ}^*$. Indeed, since $U$ is statistically independent of the pair $(\tilde{Y},Z)$, we have that $\EE[U|\tilde{Y},Z]=\EE[U]=0$, so that
\begin{align}
\EE[X|Y,Z] &= \bGa_{X\tilde{Y}}\bGa_{\tilde{Y}\tilde{Y}}^{\dag} \EE[\tilde{Y}|Y,Z] + \EE[g(Z)+U|Y,Z] \nonumber\\
&= \bGa_{X\tilde{Y}}\bGa_{\tilde{Y}\tilde{Y}}^{\dag} (Y-h(Z)) + g(Z) + \EE[U|\tilde{Y},Z] \nonumber\\
&= \bGa_{X\tilde{Y}}\bGa_{\tilde{Y}\tilde{Y}}^{\dag} (Y-h(Z)) + g(Z),
\end{align}
where we used the fact that there is a one-to-one transformation between the pair $(Y,Z)$ and the pair $(\tilde{Y},Z)$.
This expression is partially linear in $Y$, implying that this is also the PLMMSE estimator in this setting. Thus, for the distribution $F^*_{XYZ}$, the PLMMSE estimator is optimal not only among all partially linear functions, but also among \emph{all} functions of $Y$ and $Z$.

\bibliographystyle{IEEEtran}
%\bibliography{IEEEabrv,PLMMSE}


\end{document}